\documentclass[11pt]{article}
\usepackage[utf8]{inputenc}
\usepackage{vitercik}
\usepackage{float}
\usepackage{tikz}
\usepackage{indentfirst}
\usepackage{color-edits}
\usepackage{caption}
\usepackage{appendix}
\usepackage{tabularx}
\usepackage{subcaption}
\usepackage{enumitem}
\usepackage{thmtools, thm-restate}
\title{Algorithmic Contract Design for Crowdsourced Ranking}
\author{Kiriaki Frangias, Andrew Lin, Ellen Vitercik, Manolis Zampetakis}
\author{
	Kiriaki Frangias \\ UC Berkeley \\ \texttt{kfrangias@berkeley.edu} \and 
	Andrew Lin \\ Princeton \\ \texttt{andrewlin@princeton.edu}
    \and
	Ellen Vitercik \\ Stanford \\ \texttt{vitercik@stanford.edu}
  \and
	Manolis Zampetakis \\ Yale \\ \texttt{emmanouil.zampetakis@yale.edu}}
\date{}
\usepackage{graphicx}
\graphicspath{ {images/} }
\usepackage{color-edits}
 \addauthor{kf}{blue}
 \addauthor{ev}{purple}
\addauthor{al}{red}

\begin{document}

\maketitle

\begin{abstract}
    Ranking is fundamental to many areas, such as search engine optimization, human feedback for language models, as well as peer grading. Crowdsourcing, which is often used for these tasks, requires proper incentivization to ensure accurate inputs. In this work, we draw on the field of \emph{contract theory} from Economics to propose a novel mechanism that enables a \emph{principal} to accurately rank a set of items by incentivizing agents to provide pairwise comparisons of the items. Our mechanism implements these incentives by verifying a subset of each agent's comparisons, a task we assume to be costly. The agent is compensated (for example, monetarily or with class credit) based on the accuracy of these comparisons. Our mechanism achieves the following guarantees: (1) it only requires the principal to verify $O(\log s)$ comparisons, where $s$ is the total number of agents, and (2) it provably achieves higher total utility for the principal compared to ranking the items herself with no crowdsourcing.
\end{abstract}

\section{Introduction}

Ranking is a central task across a wide range of applications, including search engine optimization, human feedback for language models, and peer grading: search engines rank pages by relevance, humans rank text generation based on grammatical correctness, and peer graders rank submissions by quality. Platforms often use crowdsourcing to determine rankings \citep{10.1145/2213836.2213880,Kou2017CrowdsourcedTQ,ghose2012designing}, especially when ranking many alternatives. In many such cases, it is often more feasible and accurate for agents to provide pairwise comparisons than cardinal evaluations \citep{pmlr-v38-shah15,shah2013case}.
However, rational agents will only invest the effort necessary to provide accurate comparisons when properly incentivized.
In many real-world scenarios, these incentives involve a \emph{principal} (the ranking's end-user) verifying a selection of an agent's comparisons. Then, the agent's payment depends on performing these verified comparisons accurately. Applications where ground-truth pairwise comparisons are used to evaluate non-experts' performance include rating search engine optimization \citep{shah2016estimation}, robotics \citep{sadigh2017active}, translation \citep{kreutzer2018reliability}, peer-grading \citep{shah2013case}, reinforcement learning from human feedback \citep{rafailov2023direct}, and medical imaging \citep{orting2017crowdsourced}. 

We propose a novel mechanism that enables a principal to rank a set of items from pairwise comparisons while minimizing the requisite payments and number of verifications, which are costly for the principal. To model the principal and agents, we draw on \emph{contract theory}, a classical field in economics \citep{f7fd804f-0ce0-35f1-9381-72a4444d8df3} that has become increasingly relevant in computer science~\citep[e.g.,][]{kang2019toward,guruganesh2020contracts,duetting2021combinatorial,zhu2023sample}. The principal assigns comparisons to each agent and verifies a subset of these comparisons. The agent does not know the subset's identity.  The principal will compensate each agent for correctly evaluating all verified pairs.

Given the principal's offered payment, the agent decides whether they will exert effort. As in the standard model of contract theory, if the agent exerts effort, they will be ``good'' at the task with some probability and ``bad'' at the task with some smaller probability. If they decide not to exert effort, they will be ``bad'' with certainty. A ``good'' agent will return the correct comparisons, but a ``bad'' agent may not.
When modeling the agents, we introduce heterogeneity through differing values of the per-comparison disutility that they experience. We therefore assume that every agent can potentially be ``good'' but might require more or less compensation to exert effort. As we show, this also allows the principal to incentivize and pay only agents with small disutilities, even though these values are not \emph{a priori} known to the principal.

\paragraph{Key challenges.}
The primary obstacle we face in designing a mechanism for this problem lies in determining how to distribute the comparisons among the agents while satisfying several critical requirements:
\begin{enumerate}[label=(\alph*)]
\item The number of items assigned to each agent (i.e., the subset of items that make up their pairwise comparisons) should be small, minimizing the agent's workload.\label{item:workload}
\item Every comparison must be evaluated by at least one ``good'' agent with high probability.\label{item:redundancy}
\end{enumerate}
In other words, we want each agent's comparisons to cover as few items as possible to satisfy~\ref{item:workload}. However, we need the union of all agents' comparisons to have sufficiently high redundancy to satisfy~\ref{item:redundancy}. These conditions are fundamentally at odds, so satisfying both requires care.
\subsection{Our contributions}

We summarize our primary contributions below.
\begin{enumerate}
\item \textbf{Distributing comparisons.} We provide an algorithm for distributing comparisons among agents while ensuring conditions \ref{item:workload} and \ref{item:redundancy} above are satisfied. 
We uncover a surprising connection to a classic combinatorial problem, the \emph{social golfer problem}, which we use to design our algorithm. If there are $n$ items to be sorted and $s$ agents, each agent is assigned only $\tilde{\mathcal{O}}(n^{3/2}/s)$ items and $\tilde{\mathcal{O}}(n^{3/2}/s)$ comparisons.
\item \textbf{Aggregating the noisy comparisons.} Next, we present an algorithm for identifying ``bad'' agents and aggregating comparisons from ``good'' agents. Our algorithm returns the correct ranking with high probability.
\item \textbf{Determining payments.} We identify the payments necessary to ensure enough agents are ``good'' to recover the correct sorting. This depends on the \emph{disutility} it costs agents to make comparisons. We study a range of settings, from a warm-up where the principal knows the disutilities (as is typical in contract theory), to a setting where the disutilities are drawn from a distribution about which the principal has limited information.
\item \textbf{Principal's utility.}  We show that it suffices for the principal to verify $O(\log{s})$ comparisons to ensure that the correct sorting is recovered. We also identify natural conditions under which, by implementing our proposed contract, the principal achieves higher total utility than she would if she performed all comparisons herself.
\item \textbf{Experiments.} In experiments, we stress-test our results by relaxing our model along several axes. We demonstrate that even when the assumed model is inaccurate, with a small increase in the payment, the principal can recover the correct ranking.
\end{enumerate}

\subsection{Related work}

\paragraph{Rank Aggregation.} Ranking via pairwise comparisons can be considered a subfield of rank aggregation ~\citep{ThurstoneTheMO,Mallows1957NONNULLRM,RePEc:bla:jorssc:v:24:y:1975:i:2:p:193-202,bradley1952rank}. 
Recent work on rank aggregation with pairwise comparisons has built upon these techniques \citep{guiver2009bayesian,lu2011learning,chen2013pairwise,chen2015spectral,tsukida2011analyze, lin2010rank}.
Rank aggregation differs from our work, which models the data sources as agents with differing reliability.

\paragraph{Crowdsourcing.} 
Crowdsourcing has been commonly used in a variety of applications, including sorting \citep{chen2013pairwise,raykar2010learning,10.1145/2806416.2806451}. A central task in crowdsourcing is modeling the reliability of annotators \citep{chen2013pairwise,raykar2010learning}. \citet{chen2013pairwise}, for example, use the Bradley-Terry model for ranking aggregation, assuming that annotators are either perfect or spammers (comparisons are correct with probability 0.5). In contrast, we study a dynamic setting in which annotators can be \emph{incentivized} to exert effort and increase the accuracy of their work. 
\paragraph{Ordinal peer grading.} This area is most related to our work. With the rise of massive online open courses (MOOCs), it has become increasingly important to decentralize the grading process in an efficient manner that comes at a low cost to teaching staff \citep{piech2013tuned}. Peer grading can be a solution to this technical hurdle. Ordinal peer grading (in which students are asked to rank other students' work) has been shown to be an inexpensive, accurate, and robust way to assess performance \citep{shah2013case,barnett2003modern,carterette2008here,stewart2005absolute}, even when tested under real-world data from MOOCs
\citep{mi2015probabilistic,10.1145/2724660.2724678,caragiannis2016effective,raman2014methods}. However, many works on ordinal peer grading draw on probabilistic methods to estimate grader reliability and lack perspective on strategic interactions \citep{10.1145/2724660.2724678,raman2014methods,mi2015probabilistic}. In contrast, works relating to aggregating feedback from strategic data sources focus mostly on regression or cardinal evaluations of items \citep{jurca2006,Jurca2006RobustIF,Jurca_2009,Waggoner_Chen_2014,pmlr-v40-Cai15}. Our work aims to bridge this gap using a contract between the principal and the agents. The contract structures their interaction and leverages their strategic behavior by incentivizing agents to perform accurate comparisons. It also comes at a low cost to the principal while provably retrieving the correct ranking of items.

\section{Warm-up: Known and equal disutilities}

As a warm-up to our main result, we propose a contract assuming that (a) the principal knows the agents' disutilites and (b) every agent experiences the same value of per-comparison disutility (assumptions we relax in Section~\ref{sec:unknown}). We begin with a formal model in Section~\ref{sec:known_disutilities_model} and define the optimal payment in Section \ref{sec:agentsincentives}. In Sections \ref{sec:peergrading} and \ref{sec:allocationfofitems}, we present our main algorithm \textsc{CrowdSort} for computing the correct sorting. 
Section \ref{sec:principalsincentives} identifies conditions under which the principal gets higher utility when implementing our proposed contract than when performing all comparisons herself. The full proofs of all results in this section are in Appendix \ref{appendixa}.

\subsection{Model and notation}\label{sec:known_disutilities_model}

In our model, there is one principal, $s$ agents $ a_1, \dots, a_s$, and a set of $n$ items $T = \{T_1, \dots, T_n\}$. Each item has an unknown ground-truth score that measures some comparable feature that the principal wishes to know. Based on these scores, there is a ground-truth permutation $\mu^*: T \rightarrow [s]$ that sorts the items in ascending order.
If the ground truth score of item $T_j$ is smaller than the ground-truth score of item $T_k$, we denote this as $T_j <_{\mu^*} T_k = \mathbbm{1}\{\mu^*(T_j) < \mu^*(T_k)\}$.
The principal's goal is to find $\mu^*$ 
by leveraging the agents' ability to estimate ground-truth scores and compare items based on that estimate.  However, agents' comparisons can be unreliable.
The principal proposes a contract to each agent, in which agent $a_i$ is assigned pairs of items $C_i = \{(T_k, T_k')\}_{k=1}^{|C_i|}$ and is asked to order each pair. The principal verifies the accuracy of a small subset of these comparisons, but $a_i$ does not know which subset. Through the contract, the principal provides some fixed compensation to each agent that correctly compares all verified pairs of items. If at least one of the verified comparisons is incorrect, the agent is not paid. The principal's goal is to retrieve the ground-truth ordering $\mu^*$ while minimizing agents' payment and the number of verifications.

Each agent $a_i$ has a type $t_i \in$ \{``good'', ``bad''\}, which captures his ability to compare items. The principal does not observe $t_i$.  If the agent is ``good'', then all comparisons he returns are correct with regards to $\mu^*.$
If the agent is ``bad,'' then he chooses a uniformly random permutation $\mu \in M_T$, where $M_T$ is the set of all permutations $\mu : T \to [n]$. He then returns all comparisons to be consistent with that permutation.

An agent's type depends on his strategic behavior, which the principal does not observe. Each agent chooses whether to exert effort when comparing items and commits to this choice upfront. We denote agent $a_i$'s effort using $e_i \in \{0,1\}$. If $a_i$ exerts effort ($e_i=1$), then $\mathbb{P}[t_{i} = \text{``good''} | e_{i} = 1] = \pi$ Otherwise, $\mathbb{P}[t_{i} = \text{``good''} | e_{i} = 0] = 0$. We assume $\pi$ is known to the principal, as is standard in contract theory. 

\begin{remark}
In our experiments, we stress-test our model by examining the number of returned comparisons even when the principal’s estimate of $\pi$ is wrong. Figure 2.b shows that our algorithm is robust to a significant amount of noise in $\pi$, even when the individual value of $\pi$ differs among agents.\end{remark}

The interaction proceeds as follows:
\begin{enumerate}
    \item The principal fixes a payment $p^* \in \mathbb{R}^{+}$ that she will pay agents in order to incentivize them to exert effort.
    \item The principal verifies a subset of the comparisons and announces the size of this set to the agents.
    \item Agents choose whether to exert effort, which fixes their type upfront. 
    \item Agents assigned verified and unverified comparisons.
    \item If agent $a_i$ returns at least one incorrect verified pair, he is identified as ``bad'' and not compensated. We denote $c_{i}=\mathbbm{1}\{a_i$ identified as ``bad''$\}$ and $\hat{\pi} = \mathbb{P}[c_{i}=1 \mid t_{i}=\text{``bad''}]$. (Identifying an agent depends on their type and is independent of their effort conditioned on their type.)
    \item All agents not identified as ``bad'' receive payment $p^*$. 
\end{enumerate}
\begin{remark}
    Agents who perform verified comparisons correctly are not later penalized if conflicting, non-verified comparisons are returned by these agents. In this case, we disregard such comparisons altogether without further implications for the annotators. This allows a margin for human error without affecting agents’ payment.
\end{remark}
The principal obtains utility $u_{\mu^*} \geq 0$, which scales with the number of comparisons returned by agents not identified as ``bad''. She also experiences disutility $\bar{\psi} \geq 0$ from performing each verification.
Meanwhile, agents who exert effort experience a per-comparison disutility of $\psi \geq 0$. In this section, $\psi$ is common among all agents and known to the principal. Section~\ref{sec:unknown} relaxes these assumptions.

Appendix \ref{appendixnotation} includes a table with our notation.

\subsection{Agents' incentives} \label{sec:agentsincentives}

We first determine the smallest amount the principal can pay the agents to incentivize them to exert effort while satisfying individual rationality. We use the notation
 $\mathbb{E}[|C_i|] \leq d$ to denote a bound on the expected number of comparisons each agent $a_i$ is asked to perform.
\begin{restatable}{lemma}{payment}\label{paymentthm}\label{p^*}
    If the principal pays each agent $p^* \geq \frac{d\psi}{\hat{\pi}\pi},$ then every agent is incentivized to exert effort. 
\end{restatable}

For the remainder of this section, we set \begin{equation}p^* = \frac{d\psi}{\hat{\pi}\pi}\label{eq:p*}\end{equation} so for each agent $a_i$, $t_i =$ ``good'' with probability $\pi$. 

\subsection{Peer grading algorithm} \label{sec:peergrading}

\begin{algorithm}[t]
    \caption{\textsc{CrowdSort}($T, s, v,  r$)}\label{alg:grade}
    \begin{algorithmic}[1]
    \State Let $B= \emptyset$ 
    \State $V =\textsc{findVset}(T, v)$ \Comment{choose comparisons to verify}
    \State $\{W_i\}_{i=1}^s=\textsc{findWset}(T, s, r)$ \Comment{split $T^2$ into subsets}
    \For{$i \in [s]$}
        \State Assign $C_i = V \cup W_i$ to agent $a_i$ and obtain the set of returned comparisons $R_{V}$ and $R_{W_i}$
        \If {there exists incorrect element in $R_V$} 
            \State $B \gets B \cup \{ i\}$
	\EndIf
    \EndFor
    \State \textbf{return} $B$, $\bigcup_{i \in [s]\setminus B}(R_V \cup R_{W_i})$
    
\end{algorithmic}
\end{algorithm}
We next present Algorithm \ref{alg:grade}, which, with high probability, returns all comparisons needed to yield the ground-truth sorting. Algorithm \ref{alg:grade} takes as input the set of items $T$, the number of agents $s$, the number of pairs of items $v = |V|$ to be verified by the principal,
and the number of agents $r$ that each comparison is assigned to. Algorithm \ref{alg:grade} partitions all $\binom{n}{2}$ pairs of items into disjoint subsets $V, \{W_i\}_{i=1}^s \subseteq [n]\times [n]$, by calling algorithms \textsc{findVset} and \textsc{findWset} respectively. We define $V$ as the set of pairwise comparisons that the principal verifies and $W_i$ as the set of unverified comparisons assigned to agent $a_i$. Algorithm \ref{alg:grade} assigns to agent $a_i$ the set $C_i = V \cup W_i$. Each agent returns sets $R_V$ and $R_{W_i}$, which are the pairwise comparisons of the sets $V$ and $W_i$, respectively. If the agent returns an incorrect verified comparison, Algorithm \ref{alg:grade} adds the agent to the set $B$ of ``bad'' agents that have been identified. The algorithm then returns $B$ along with all comparisons returned by agents not identified as ``bad''.

\subsection{Allocation of items} \label{sec:allocationfofitems}

In this section, we describe our algorithms for assigning verified and unverified pairs of items to agents (\textsc{findVset} and \textsc{findWset}, respectively) and show that: (1) for every verified pairwise comparison $(T_i, T_j) \in V$ assigned to a ``bad'' agent, the agent outputs $(T_i<_{\mu} T_j) = 1$ with probability $1/2$, and (2) a ``bad'' agent performs each verified pairwise comparison independently from all others.
In Lemma \ref{ver}, we use these two properties to prove that with high probability, Algorithm \ref{alg:grade} identifies all ``bad'' agents.

\subsubsection{Assigning verified comparisons.}

To guarantee that a ``bad'' agent will be identified with probability $1-1/2^{|V|}$, it suffices to create the sets $V$ and $\{W_i\}_{i=1}^s$ such that each ``bad'' agent can only sample the output of each verified comparison independently from Bernoulli$(1/2)$. We say a verified comparison is \emph{redundant} if a ``bad'' agent does not sample the result of the comparison from Bernoulli($1/2$) conditioned on the realization of all other comparisons in $V$.

\begin{definition}
    Let 
    $E_{i,m, m'}$ be the event that agent $a_i$ with $t_i =$ ``bad'' reports that the ground truth score of item $T_m$ is smaller than that of item $T_{m'}$. 
    We say that $C_i \subseteq T \times T$ contains a \emph{redundant pair} $(T_{j}, T_{j'})$ if: $\mathbb{P} \left[ E_{i,j, j'} \mid E_{i,k, k'}  \; \forall (T_{k}, T_{k'}) \in V \setminus \{(T_{j}, T_{j'})\} \right] \neq \frac{1}{2}.$
\end{definition}

\begin{restatable}{lemma}{redundant}\label{redundant}
    If each item in $V$ appears in exactly one pair, then no pair in $V$ is redundant.
\end{restatable} 
\begin{proof}[Proof sketch]
We prove that when verified pairs are disjoint, the number of permutations assigning $(T_j <_{\mu} T_j')$ conditioned on all other verified comparisons is exactly $|M_T|/2$. Therefore, no information can be inferred about it.
\end{proof}

Based on Lemma~\ref{redundant}, \textsc{findVset}$(T, v)$ chooses $2v$ arbitrary items and arbitrarily pairs them.
 We next bound the number of verifications sufficient to catch all ``bad'' agents.

\begin{restatable}{lemma}{ver}\label{ver}
If each item in $V$ appears in no more than one pair and $|V| \geq \log \left(2(1-\pi)s/ \delta \right)$, then with probability $1-\frac{\delta}{2}$ over $t_1,...,t_s \sim \textnormal{Bernoulli}(\pi)$, \textsc{CrowdSort}($T, s, |V|, r$) returns  $B = \{i \mid \forall a_i \; \text{s.t.} \; t_i =\text{``bad''}\}.$
\end{restatable}

\subsubsection{Assigning unverified comparisons.}

We now define \textsc{findWset}.
First, we bound the number of agents each comparison must be assigned to. 

\begin{restatable}{lemma}{redundancy}\label{redundancy}
Suppose $v \geq \log \left(2(1-\pi)s/ \delta \right)$. Let $r \geq \frac{1}{\log(1-\pi)}\log \left( \frac{\delta}{2n^2}\right)$ be the number of agents each comparison is assigned to. Then with probability $1-\delta$, \textsc{CrowdSort}($T, s, v, r$) returns the ground-truth ordering.
\end{restatable}

Based on this lemma, we define \begin{equation}r^* = \frac{\log \left( \frac{\delta}{3n^2}\right)}{\log (1-\pi)} \quad \text{ and } \quad v^* = \log \left( \frac{2(1-\pi)s}{ \delta} \right).\label{eq:opt_r_v}\end{equation}

\smallskip
\emph{Algorithm overview.}
\begin{algorithm}[t]

    \caption{\textsc{findWset}($T, s, r$)}\label{alg:findWset}
    \begin{algorithmic}[1]

    \State $\{W_i\}_{i=1}^{s} = \emptyset^s$
    \State $q =$
    \textsc{findPrime}$(\sqrt{|T|})$
    \State Generate a set $\tilde{T}$ of hallucinated items s.t. $|\tilde{T}| = q^2-T$
    \State $\{P_0, \dots, P_{q-1}\} =$ SGP$(T \cup \tilde{T})$
    \State Arbitrarily split agents into $G_0, \dots, G_{s/r}$ groups
    \State $k = 0$
    \For{$\ell = 0, \dots, q-1$} \Comment{iterate over partitions}
        \For{$c=0, \dots, q/r$} \Comment{iterate over subsets}
            \For{$j$ s.t. $a_j \in G_k$}
                \State $W_j \leftarrow \{ (T_i, T_j) \; | \; T_i, T_j \in T \wedge (T_i, T_j) \in (P_\ell)_{c(rn/s)+n} \times (P_\ell)_{c(rn/s)+n}\}_{n=0}^{rn/s-1}$ \Comment{assign to $rn/s$ agents}
            \EndFor
            \State $k \leftarrow k+1$ 
        \EndFor
    \EndFor
    \State Return $(\{W_1, \cdots, W_s\})$
\end{algorithmic}
\end{algorithm}
\textsc{findWset} (Algorithm~\ref{alg:findWset})
assigns each agent $rn/s$ disjoint subsets of items, each containing approximately $\sqrt{n}$ items. The agent is asked for all comparisons within each subset---or in other words, a total ordering of each subset. (Note that to compute a total ordering of a subset, the agent only needs to manually perform $\tilde{O}(\sqrt{n})$ comparisons, leading to a total of $\tilde{O}(r n^{3/2}/s)$ comparisons.)

To create these subsets,
\textsc{findWset} calls the algorithm \textsc{SGP}, which is from the literature on a classic combinatorial problem, the \emph{social golfer problem}. We include algorithm \textsc{SGP} for completeness in Appendix~\ref{appendixa} (Algorithm~\ref{alg:sgp}). \textsc{SGP} generates approximately $\sqrt{n}$ partitions of the items, where each partition consists of approximately $\sqrt{n}$ disjoint subsets of equal size, such that every pair of elements are in the same subset in exactly one partition (Theorem \ref{thm:sgp}).
\textsc{findWset} then round-robin assigns each agent $rn/s$ subsets.

Connecting to the SGP guarantees that each agent will be assigned as few as $O(rn^{3/2}/s)$ items. At the same time, \textsc{SGP} distributes items into subsets such that every pair of items appears in the same subset \emph{exactly once}. This allows the principal to assign each comparison to exactly $r$ agents.

\smallskip
\emph{Details about the social golfer problem.} In the social golfer problem, the goal is to determine whether $m$ golfers can play together in $g$ groups of $p$ (where $m = gp$) for $w$ days in such a way that (1) no two golfers play in the same group more than once, and (2) each golfer plays exactly once each day. Thus, an instance of this problem is defined by integers $(g,p,w).$ For us, the golfers correspond to the items, the days correspond to the partitions, and the groups of golfers correspond to the sets in the partitions. Thus, we aim to solve this problem with $g$, $p$, and $w$ all approximately equal to $\sqrt{n}.$

\smallskip
\emph{Technical details about the algorithm.} A technical hurdle we face is that \textsc{SGP} only returns a solution when the square root of the number of elements is prime (see Theorem \ref{thm:sgp} in Appendix~\ref{appendixa} \citep{SGP}). \textsc{findWset} therefore calls \textsc{findPrime}$(\sqrt{n})$, which returns the smallest prime $q$ such that $q \geq \sqrt{n}$. Then \textsc{findWset} generates a set $\tilde{T}$ of $q^2-n$ hallucinated items which, along with the set of existing items $T$, are passed to \textsc{SGP}. Hallucinated items are used as placeholders and do not affect our procedure or contract because agents are not asked to sort them.  Assuming Cramer's conjecture holds (a well-known conjecture which states that the difference between consecutive primes grows poly-logarithmically), $q = \tilde{O}(\sqrt{n}),$ so this complication comes at minimal cost.
We now bound the number of items and comparisons assigned to each agent. 
\begin{restatable}{lemma}{numcomps} \label{number of comparisons}
Assuming Cramer's conjecture, each agent is assigned $\tilde{\mathcal{O}}(n^{3/2} / s)$ items and $\tilde{\mathcal{O}}(n^{3/2} / s)$ comparisons. 
\end{restatable}

\subsection{Principal's incentives} \label{sec:principalsincentives}

We synthesize our results in the following theorem, in which we additionally prove that if the principal's disutility $\bar{\psi}$ is sufficiently larger than the agents' disuility $\psi$, then it is in the principal's best interest to use our mechanism rather than perform all comparisons herself. This is a natural assumption, which implies that the principal would have to be compensated sufficiently more per comparison than the agents for their work to be profitable to the principal.

\begin{restatable}{theorem}{bigthm} \label{bigthm}
Suppose that the principal pays agents $p^{*}$ and runs \textsc{CrowdSort($T, s, v^{*}, r^{*}$)}. Then:
\begin{enumerate}
    \itemsep0em 
    \item Assuming Cramer's conjecture, each agent performs $\tilde{\mathcal{O}}(n^{3/2}/s)$ comparisons.
\item  The principal verifies $\mathcal{O}(\log{(s/\delta)})$ comparisons.
\item With probability $1-\delta$, the correct ranking is retrieved.
\item Assume that $\delta < 1/s^2$ and $\pi$ is bounded away from $1$. If $\bar{\psi} /\psi = \Omega(\sqrt{n}\log{n})$,  then the principal's expected utility is at least the utility she would get if she were to perform all comparisons on her own.
\end{enumerate}

\end{restatable}

\section{Agents with unknown disutilites}\label{sec:unknown}
We now significantly expand our model to study heterogeneous agents with differing, unknown disutilities. The full proofs of all results in this section are in Appendix \ref{appendixb}.

\subsection{Model and notation}

In this section, agent $a_i$ experiences per-comparison disutility $\psi_i \in \mathbb{R}^{+}$ while exerting effort. By allowing for different levels of disutility, we capture the differing levels of difficulty agents may face in comparing items. We assume that $\psi_1, \dots, \psi_s$ are drawn i.i.d. from some distribution $\mathcal{D}$ with CDF $F$. The realized disutility of each agent is known to the agent but unknown to the principal and all other agents. The principal has access to an estimate of the distribution $\mathcal{D}$ by way of a function $\hat{F}(\psi)$ such that $\lVert\hat{F}-F\rVert_{\infty}\leq\epsilon$. The function $\hat{F}$ could be the empirical CDF based on $O(1/\epsilon^2)$ samples, for example, as per the Dvoretzky–Kiefer–Wolfowitz inequality. Without loss of generality, we assume $\psi_1 < \psi_2 < \cdots < \psi_s$ and make no further distributional assumptions. 

Some agents might experience high disutility when exerting effort, so it may not be in the principal's best interest to incentivize all agents to exert effort. 
For the rest of this section, we denote by $g$ the number of agents that the principal aims to incentivize; in Theorem \ref{g*}, we describe an efficient optimization problem that solves for the optimal value of $g$, denoted $g^*$. The principal controls the number of incentivized agents only through the payment function. Hence, the realized number of agents that exert effort, denoted $\hat{g}$, is a random variable with a distribution that depends on the payment. Under this model, the principal will:
\begin{enumerate}
    \itemsep0em 
    \item Calculate $g^*$, the optimal number incentivized agents.
    \item Announce the payment $p_{g^*} \in \mathbb{R}^{+}$ (based on $g^*$).
    \item Verify a subset of comparisons assigned to the agents.
    \item Assign each comparison to a total of $r$  agents.
    \item Pay $p_{g^*}$ to all agents not identified as ``bad.''
\end{enumerate}

\subsection{Payment function}
 The following analysis provides a lower bound on the payment $p_g$ required so that approximately $g$ agents exert effort while satisfying individual rationality. We denote the payment function as $p_i(c_i) = p_g \mathbbm{1}\{c_i \neq 1\}$.
\begin{restatable}{lemma}{paymentunk}\label{paymentunk}
    Suppose $p_g \geq \frac{d\psi_i}{\hat{\pi}\pi}$, where $d$ denotes an upper bound on the expected number of comparisons performed by an agent. Then agent $a_i$ is incentivized to exert effort. 
\end{restatable}

Given $d$, the number of agents incentivized to exert effort depends solely on the disutilities $\psi_i$ of each agent, so the approximate distribution of the number of agents exerting effort is known to the principal and can be used to determine the payment function. In the following lemma, we upper bound the probability that an agent exerts effort, which implies an upper bound on the total payment.
In the following lemma, we denote $\tilde{F}$ to be the quantile function with respect to $F$, i.e $\tilde{F}(p) = \inf\{x \in \mathbb{R} : p\leq F(x)\}$.

\begin{restatable}{lemma}{numincentivizedagents} \label{num incentivized agents}
Suppose the principal decides to incentivize $g$ agents. If the principal sets $p_g \leq \frac{d}{\hat{\pi}\pi}\tilde{F}_{N}\left(\frac{g}{s} \right)$, then the probability that any agent will exert effort is at most $g/s +\epsilon$, or in other words, $\Pr_{\psi_i \sim \mathcal{D}}[\psi_i \geq p_g] \leq g/s +\epsilon.$
\end{restatable}

For the remainder of this section, we set \begin{equation}
p_g = \frac{d}{\hat{\pi}\pi}\tilde{F}_{N}\left(\frac{g}{s}\right).
\label{payment}\end{equation}

\begin{remark}\label{ub}
    Notice that if Equation~\eqref{payment} holds, the probability that any agent exerts effort is $\frac{g}{s}-\epsilon  \leq F\left(\frac{p_g\hat{\pi}\pi}{d}\right) \leq \frac{g}{s}+\epsilon.$
\end{remark}

\subsection{Correctness of \textsc{CrowdSort}}
After choosing the payment function $p_g$ the principal will choose to verify a subset of the binary comparisons, $V$. The size of this set $v_g = |V|$ must be such that with high probability, all agents of type ``bad'' are identified by algorithm \textsc{CrowdSort}($T, s, v_g, r_g$). The following lemma provides a sufficient lower bound on the value of $v_g$.

\begin{restatable}{lemma}{versecond} \label{ver2}
Suppose the principal chooses to verify a set of pairs of items, V, such that each item in V appears in no more than one pair. Let $v_g \geq \log_2 \left( 2(s-\pi g +\pi s \epsilon)/\delta \right)$ and $p_{g}$ be the payment function, set according to Equation~\eqref{payment}. Then with probability $1-\frac{\delta}{2}$ over the types of the agents, $t_1,...,t_s$, \textsc{CrowdSort}($T, s, v_g, r_g$) returns  $B = \{i | t_i =$``bad''$\}$. 
\end{restatable}

Next, we bound the redundancy required to ensure that \textsc{CrowdSort} returns the ground-truth ordering.
\begin{restatable}{lemma}{redundancysecond} \label{redundancy2}
Assume $v_g \geq \log_2 \left(2(s-\pi g +\pi s \epsilon)/\delta \right)$ and let $r_g \geq \frac{\log_2 \left( \delta/(2n^2)\right)}{\log_2 \left( 1-\pi g/s +\pi \epsilon\right)}$ be the number of agents each comparison is assigned to. With probability $1-\delta$, \textsc{CrowdSort}($T, s, v_g, r_g$) returns the ground-truth sorting.
\end{restatable}

\subsection{Principal's incentives}
We combine the above results in the following theorem, in which we additionally prove that it is possible to efficiently compute the optimal number of agents to incentivize.

\begin{restatable}{theorem}{optg} \label{g*}
Suppose that the principal pays agents $p_g$ and runs \textsc{CrowdSort($T, s, v_g, r_g$)}, with $v_g$ and $r_g$ bounded as in Lemma~\ref{redundancy2}, and a payment of $p_g$. Then:
\begin{enumerate}
\itemsep0em 
\item Assuming Cramer's conjecture holds, each agent performs $\tilde{\mathcal{O}}(n^{3/2} / s)$ comparisons.
\item The principal verifies $\mathcal{O}(\log{(s/\delta)})$ comparisons.
\item With probability $1-\delta$, the correct ranking is retrieved.
\item In $O(s)$ time, it is possible to compute the value $g^* \in \{0, \dots, s\}$ such that the principal's expected utility is maximized compared to all choices of $g \in [s]$.
\end{enumerate}
\end{restatable}

We show empirically that for $\bar{\psi}/\psi$ sufficiently large, the principal should run \textsc{CrowdSort}, not sort all items herself.

\section{Experiments}\label{sec:experiments} 
We present experiments that test our models and algorithm. We evaluate the impact of our algorithm on the principal's utility compared to when she does not incentivize agents, the efficacy of our algorithm in returning correct comparisons, and the robustness of our model to noisy values of $\psi$ and $\pi$.

\smallskip \emph{Experimental procedure.} Unless otherwise specified, we keep the following set of hyperparameter values across our experiments: $\pi = 0.8$, $\delta = 0.01$, $s = 100$, $n = 100$, $\bar{\psi} = 2$, and $ \psi = 0.01$. 
Peer grading is an example of a setting where $s = n$ since the number of assignments is the same as the number of students.
In Appendix~\ref{appendixexperiments}, we include experiments with a broad variety of parameter settings. 

Since we add noise to our model in these experiments, \textsc{CrowdSort} may not catch all bad agents and thus may return conflicting comparisons. We discard all conflicting comparisons (see Algorithm~\ref{postprocess} in Appendix~\ref{appendixexperiments}). Despite this pessimistic post-processing strategy, our experiments show that we can recover most comparisons even under noise.

\subsection{Principal's utility}

 \begin{figure}[t]
 \centering
     \begin{subfigure}[b]{0.4 \textwidth}
         \includegraphics[width=\textwidth]{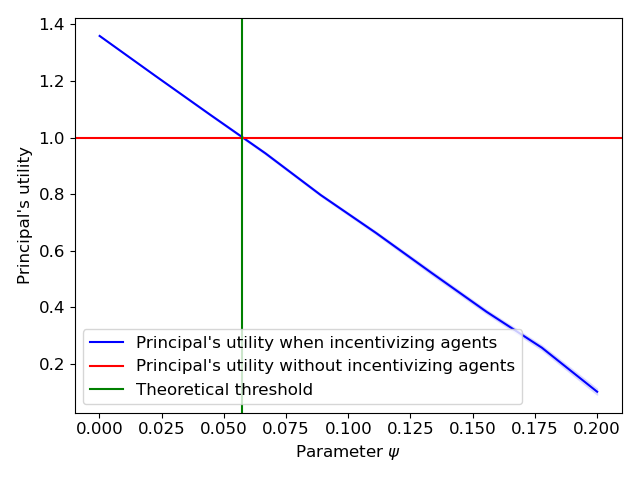}
         \centering
         \caption{No noise added to $\psi$. 
         }
         \label{fig:utilityvspsi}
     \end{subfigure}
     \hfill
     \begin{subfigure}[b]{0.4\textwidth}
         \includegraphics[width=\textwidth]{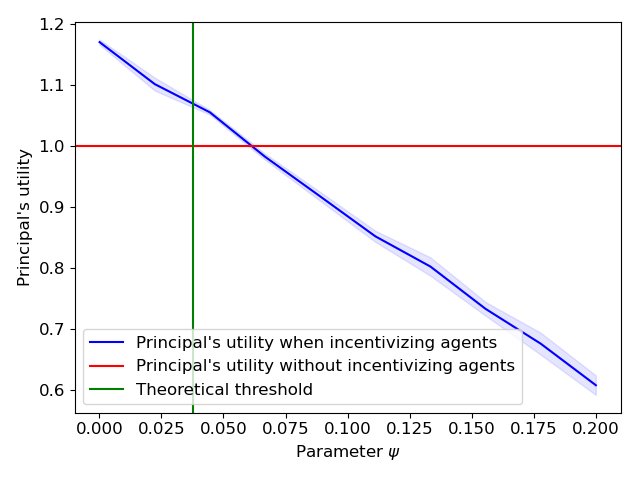}
         \centering
         \caption{Noise added to $\psi$. 
         }
         \label{fig:approxutilityvspi}
     \end{subfigure}
     \caption{Principal's utility as a function of $\psi$.\label{fig:utility}}
 \end{figure}

In Figure \ref{fig:utilityvspsi}, we plot the principal's utility as a function of $\psi$ (in Appendix \ref{appendixexperiments}, we plot utility as $\pi$ varies). The blue line shows the principal's utility when she runs \textsc{CrowdSort} with $p^*$, $v^*$, and $r^*$ defined in Equations~\eqref{eq:p*} and \eqref{eq:opt_r_v} (we describe this calculation in more detail in the next paragraph). The blue shading is a 90\% confidence interval across 50 trials. The red line shows the principal's utility if she performed all comparisons herself. Finally, the green line shows the theoretical threshold when $s=n$ is $\bar{\psi} /\psi = \Omega(\sqrt{s}\log{s})$ from Theorem~\ref{bigthm} (we use the value of this threshold without big-Omega notation: Equation~\eqref{condition} from the proof). This indicates the value of $\psi$ for 
 which the principal is indifferent between running \textsc{CrowdSort} or performing all comparisons herself.

The empirical utilities are calculated as follows.
Let $m$ be the number of agents not identified as ``bad'' by \textsc{CrowdSort} and $k$ be the number of comparisons after post-processing. We set $\lambda = 2$ to be the principal's per-comparison utility and define $U_{p^*} = \lambda k - \bar{\psi}|V|-p^{*}m$
to be her empirical utility for running \textsc{CrowdSort} with $p^*$, $v^*$, and $r^*$ defined in Equations~\eqref{eq:p*} and \eqref{eq:opt_r_v}. Similarly,
$U_0 = \lambda \binom{n}{2} - 2\bar{\psi} n\log n$
is her expected utility for performing all comparisons herself, where $2n\log n$ is the expected number of comparisons required by quicksort (this value could easily be adjusted for any sorting algorithm). For ease of comparison, the y-axis is scaled by dividing all values by the baseline utility $U_0$.

 In Figure \ref{fig:approxutilityvspi}, we add noise to the parameter $\psi$ and analyze the effect on the principal's utility. We sample the disutility of each agent independently and identically from the distribution Uniform($\psi-0.5\psi, \psi+0.5\psi$). However, we still calculate $p^*$, $v^*$, and $r^*$ from Equations~\eqref{eq:p*} and \eqref{eq:opt_r_v} using $\psi$ itself instead of the noisy version. The blue shading corresponds to a 90\% confidence interval across 50 trials.

 \smallskip
 \paragraph{Discussion.} Comparing Figures~\ref{fig:utilityvspsi} and \ref{fig:approxutilityvspi}, we can see that even when noise is added to our model, our theoretical threshold is close to the point at which the principal's utility from running \textsc{CrowdSort} dips below their utility from performing all comparisons herself. This demonstrates that our theoretical results are robust to noise in the model.

When adding noise to $\psi$, the principal must pay fewer agents because fewer are incentivized to exert effort. Thus, we see a steeper decline in Figure \ref{fig:utilityvspsi} compared to \ref{fig:approxutilityvspi}.

\subsection{Number of pairwise comparisons returned} \label{sec:experimentsnumcomps}

We evaluate the ability of \textsc{CrowdSort} to return enough correct comparisons to recover the ground-truth sorting after post-processing.
\begin{figure}[t]
\centering
    \begin{subfigure}[b]{0.4 \textwidth}
        \includegraphics[width=\textwidth]{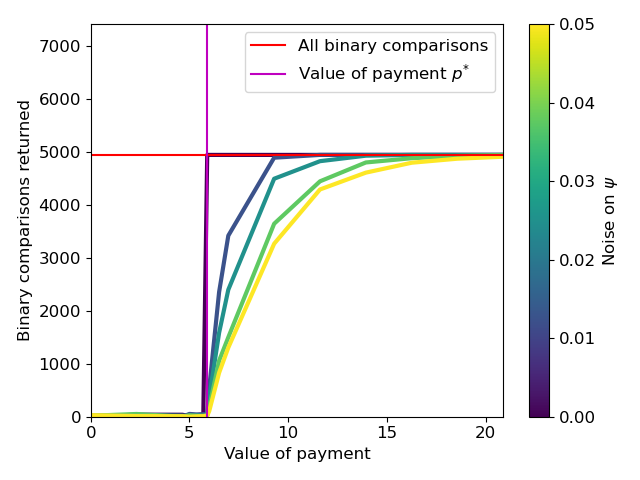}
        \centering
        \caption{Effect of noise in  parameter $\psi$ on the number of comparisons returned as a function of the payment.}
        \label{fig:compsvspsi}
    \end{subfigure}
    \hfill
    \begin{subfigure}[b]{0.4\textwidth}
        \includegraphics[width=\textwidth]{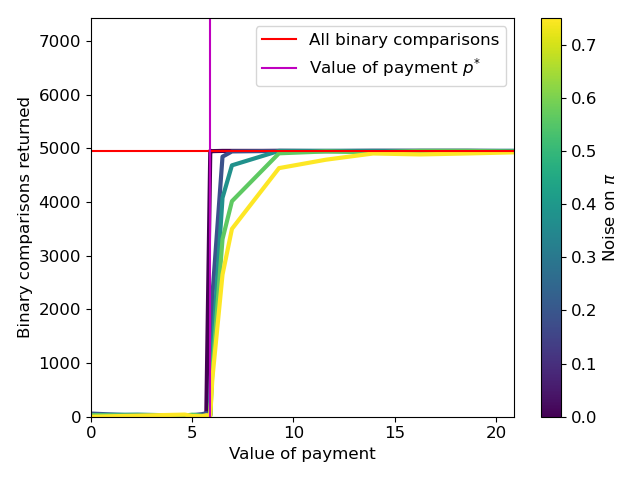}
        \centering
      \caption{Effect of noise in  parameter $\pi$ on the number of comparisons returned as a function of the payment.}
        \label{fig:compsvspi}
    \end{subfigure}
    \caption{Plots showing the number of comparisons returned as a function of $\pi$ and $\psi$. The heat plots indicate the support of uniformly sampled noise added to $\psi$ and $\pi$.\label{fig:comps}}
\end{figure}
Figures \ref{fig:compsvspsi} and \ref{fig:compsvspi} show the number of pairwise comparisons the algorithm returns as a function of the payment given to agents not identified as ``bad''. The horizontal lines  equal $\binom{n}{2} = \binom{100}{2} = 4950$. The vertical line equals the payment $p^{*}$ from Equation~\eqref{eq:p*}, which is a lower bound, above which agents are incentivized to exert effort. 

The multi-colored plots in Figure~\ref{fig:compsvspsi} indicate the number of comparisons returned when noise is added to $\psi$.
Given a value $H$ from the heatmap, each agent's disutility equals $\psi + \eta_i$ where $\psi = 0.01$ and $\eta_i \sim \text{Uniform}(0, H)$. Lighter-colored plots correspond to larger values of $H$, which means fewer binary comparisons are returned. Each plot indicates the average number of comparisons returned across 50 trials.

Similarly, the multi-colored plots in Figure~\ref{fig:compsvspi} indicate the number of comparisons returned when noise is added to $\pi$.  Given a value $H$ from the heatmap, each agent's probability $\pi_i = \Pr[t_i = \textnormal{``good''} \mid e_i = 1]$ equals $\pi - \eta_i$ where $\pi = 0.8$ and $\eta_i \sim \text{Uniform}(0, H).$ Each plot indicates the average number of comparisons returned across 10 iterations.

 \smallskip \emph{Discussion.}
These experiments show that our algorithm retrieves all comparisons in the absence of noise. 
Moreover, even when a large amount of noise is added to $\psi$ or $\pi$, our algorithm returns almost all of the comparisons with sufficient payment, despite the post-processing phase 
excluding conflicting comparisons, illustrating our algorithm's robustness.

\subsection{Optimization problem} \label{sec:experimentsopt}
In Section \ref{sec:unknown}, we relaxed the assumption that the principal knows the agents' disutilities by letting each agent's disutility be i.i.d. drawn from a distribution $\mathcal{D}$ the principal only has sample access to.
\begin{figure}[t]
   \centering
  \includegraphics[width=0.4\textwidth]{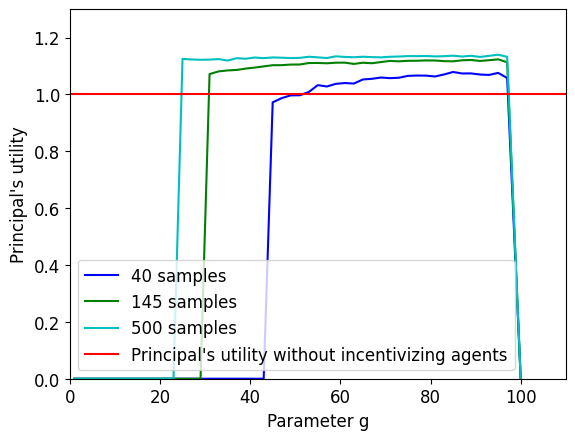}
   \caption{Principal's utility as a function of the number of agents incentivized $g$. 
    }
   \label{fig:opt}
\end{figure}
In Figure \ref{fig:opt}, we plot the principal's utility as a function of the number of agents incentivized. The horizontal line denotes the principal's utility if she were to perform all comparisons herself. The multicolored plots indicate the principal's utility for different sample sizes (namely, 40, 145, and 500) drawn from $\mathcal{D}$, which in this case is $\mathcal{N}(0.03, 0.01)$. 
We use parameter values $\bar{\psi} = \lambda = 3$. Each line indicates the average number of comparisons returned across 50 trials. Notice that many values of $g$ cannot be used as they cause the input to logarithms (e.g., defining $r^*$) to be negative.

 \smallskip \emph{Discussion.}
For suitable values of $g$, approximately $40$ samples are required for the principal's utility under our model to surpass the utility she would get if she were to perform all comparisons herself. The principal's utility increases as we increase either $g$ or the number of samples. This result is shown in the multi-colored plots in Figure \ref{fig:opt}. Therefore, even for a small number of samples, the principal's utility when implementing \textsc{CrowdSort} is well above her utility if she performed all comparisons on her own. 

\section{Conclusions and future directions}
We studied the multifaceted problem of designing a crowdsourcing mechanism that efficiently and accurately ranks a set of items using pairwise comparisons from rational agents. We based our approach on the classic principal-agent model from contract theory. To distribute comparisons among agents, we used an unexpected connection to the social golfer problem which allowed us to simultaneously minimize the agents' workload and ensure enough agents evaluate each comparison. We showed that by optimizing the payment mechanisms, the principal can incentivize enough agents to exert effort, ensuring our algorithm recovers the ground-truth ranking. Our experiments showed that even when noise is added to our original model, our method consistently aids the principal in retrieving the accurate ground-truth ordering. 

Our results open up a variety of questions for future research. We assumed agents cannot collude, which allowed us to assign the same set of verified comparisons to all agents. How should our mechanism change if agents can collude? Moreover, we assumed that agents who do not exert effort will return comparisons consistent with a uniformly-random permutation of the items. What other models of agent behavior can we study under this framework? Finally, as is standard in contract theory, we assumed that the principal knows the probability $\pi$ that an effortful agent is good. How should we handle the case where $\pi$ is unknown?

\newpage 
\bibliographystyle{plainnat}
\bibliography{refs}

\newpage
\onecolumn
\begin{appendices}

\section{Summary of notation} \label{appendixnotation}

\renewcommand{\arraystretch}{1.2}
\begin{table}[h]\normalsize
    \centering
    \caption{Summary of Most Commonly Used Notation}
    \begin{tabular}{p{0.25\linewidth} | p{0.7\linewidth}}
        \hline
        \textbf{Notation} & \textbf{Definition}    \\
        \hline
        $a_1, \cdots, a_s$ &  $s$ agents\\
        \hline
        $T = \{T_1, \cdots, T_n\}$ & set of $n$ items\\
        \hline
        $\mu^{*}: T \rightarrow [n]$ & ground-truth permutation \\
        \hline
        $M_T$ & set of all permutations $\mu: T \rightarrow [n]$\\
        \hline
        $V \subseteq [n] \times [n]$ & verified set of pairs of items  \\
        \hline
        $W= \{W_i\}_{i=1}^{s}$   & collection of subsets of pairs of items assigned to $s$ agents \\
        \hline
        $t_i \in \{\text{``good'', ``bad''}\}$ & type of agent $a_i$ \\
        \hline
        $e_i \in \{0, 1\}$ & indicates whether agent $a_i$ exerts effort \\
        \hline
        $\pi \in [0, 1]$ & $\mathbb{P}[t_i = \text{``good''} | e_i = 1]$\\
        \hline
        $c_i \in \{0, 1\}$ & indicates whether agent $a_i$ is identified as ``bad'' \\
        \hline
        $\hat{\pi} \in [0, 1]$ & $\mathbb{P}[c_i = 1 | t_i = \text{``bad''}]$\\
        \hline
        $\psi \in \mathbb{R}^{+}$ & (Section 2) per-comparison disutility of agents \\
        \hline
        $\bar{\psi}$ & per-comparison disutility of the principal \\
        \hline
        $\lambda$ & per-comparison utility of the principal \\
        \hline
        $d$ & upper bound on the expected number of pairwise comparisons each agent performs \\
        \hline
        $T_j <_{\mu} T_k$ & indicates whether the ground truth score of $T_j$ under $\mu$ is smaller than that of $T_k$ \\
        \hline
        $F, \hat{F}, \tilde{F}$ & CDF, estimate of CDF, quantile function corresponding to distribution $\mathcal{D}$ over disutilities \\
        \hline 
        $\hat{g}, g^*$ & realized number and utility-maximizing number of incentivized agents \\
        \hline
    \end{tabular}
\end{table}

\section{Omitted proofs from Section 2} \label{appendixa}
\payment*
\begin{proof}    
Let $\mathbb{E}[|C_i|] \leq d$ for some $d > 0$ be an upper bound on the expected number of comparisons that each agent is asked to perform. We denote agent $a_i$'s utility as $u_i$ and calculate its expected value given that the agent exerts effort:
\begin{align*}
     \mathbb{E}[u_i| e_{i}=1] & = \mathbb{E}[p(c_i) - |  C_i|\psi \cdot e_{i}|e_{i}=1] 
    \\ & = p^{*} \mathbb{P}[p(c_i)=p^{*}|e_{i}=1] - \mathbb{E}[|C_i|] \psi
\end{align*} 
Since $p_i = p^*$ if either (1) $t_i$ = ``good'' or (2) $t_i$ = ``bad'' but the agent is not identified as ``bad'' ($c_i = 0$), we have:
\begin{align*}
    \mathbb{E}[u_i| e_{i}=1]  & = p^{*} \left( \mathbb{P}[t_{i}=\text{``good''}|e_{i}=1]\right.
    \\ &  \left. +\mathbb{P}[t_{i}=\text{``bad''} \cap c_{i}=0 |e_{i}=1] \right) - \mathbb{E}[|C_i|]\psi 
\end{align*}
With Bayes' rule and simplifying we get:
\begin{align*}
    & \mathbb{E}[u_i| e_{i}=1] 
    \\ & = p^{*}\pi + p^{*}\mathbb{P}[c_{i}=0|t_{i} =\text{``bad''}, e_i=1]\mathbb{P}[t_{i} = \text{``bad''}|e_i=1] - \mathbb{E}[|C_i|]\psi \\
        & = p^{*}\pi + p^{*}(1-\hat{\pi})(1-\pi) - \mathbb{E}[|C_i|]\psi
\end{align*} 
We similarly handle the second case, in which the agent does not exert effort, this time noting that the agent does not experience disutility:
\begin{align*}
    \mathbb{E}[u_i| e_{i}=0] & = \mathbb{E}[p(c_i)  - |C_i|\psi \cdot e_{i}|e_{i}=0] \\ & = p^{*}\mathbb{P}[p(c_i)=p^{*}|e_{i}=0]  \\
    & = p^{*}\mathbb{P}[c_{i}=0|t_{i} =\text{``bad''}] \\
    & = p^{*}(1-\hat{\pi})
\end{align*} 
For the agent to be incentivized to exert effort, the payment needs to be such that:
\begin{align*}
    \mathbb{E}[u_i|e_i =1]  \geq \mathbb{E}[u_i|e_i =0].
\end{align*} 
We use the quantities from above to get:
\begin{align*}
     p^{*}\pi + p^{*}(1-\hat{\pi})(1-\pi) - \mathbb{E}[|C_i|]\psi \geq p^{*}(1-\hat{\pi})
\end{align*}
which holds if and only if $p^* \hat{\pi}\pi \geq \mathbb{E}[|C_i|]\psi$. By $\mathbb{E}[|C_i|] \leq d$ and simplifying we get the lower bound:
\begin{align} \label{constraint}
    p^* \geq \frac{d\psi}{\hat{\pi}\pi}
\end{align}
To satisfy individual rationality, we require that:
\begin{align*}
    &\mathbb{E}[u_i| e_{i}=1] \geq 0
\end{align*}
which holds if and only if:
\begin{align*}
    p^{*}\pi + p^{*}(1-\hat{\pi})(1-\pi) - \mathbb{E}[|C_i|]\psi \geq 0
\end{align*}
and by simplifying we get:
\begin{align*}
p^* \geq \frac{\mathbb{E}[|C_i|]\psi}{1+\hat{\pi}\pi 
 - \hat{\pi}} 
\end{align*}
which is always satisfied if constraint \ref{constraint} is satisfied.
\end{proof}
\redundant*
\begin{proof}
    Let $\tilde{\mu} \in M_T$ be the permutation that an arbitrary ``bad'' agent $a_i$ chooses. We must show that for every pair of items $(T_{V_j}, T_{V_j'})\in V$:
    \begin{align*}
        \mathbb{P} \left[ T_{V_j}<_i T_{V_j'} \mid T_{V_k} <_i T_{V_k'}  \quad \forall (T_{V_k}, T_{V_k'}) \in V \setminus \{(T_{V_j}, T_{V_j'})\} \right] = \frac{1}{2}
    \end{align*}
    or equivalently:
    \begin{align*}
        \mathbb{P} \left[ \underbrace{T_{V_j}<_{\tilde{\mu}} T_{V_j'}}_{P_1} \mid \underbrace{T_{V_k} <_{\tilde{\mu}} T_{V_k'}  \quad \forall (T_{V_k}, T_{V_k'}) \in V \setminus \{(T_{V_j}, T_{V_j'})\} }_{P_2}\right] = \frac{1}{2}
    \end{align*}
    This holds if and only if the number of permutations that satisfy property $P_1$ and property $P_2$ is exactly half of the number of permutations that only satisfy property $P_2$. This holds since no items in the set V appear in a pair more than once so $(T_{V_j}, T_{V_{j'}})$ is not contained in $P_2$. Therefore, no information can be inferred about the relative indices of items $T_{V_j}$ and $T_{V_j'}$.
\end{proof}

\ver*

\begin{proof}
    By Lemma~\ref{redundant} we have that each agent of type ``bad'' samples every verified comparison i.i.d. from Bernoulli(1/2). Therefore, we have that for every ``bad'' agent $a_i$:
    \begin{align*}
        \mathbb{P}[i \not \in B \mid t_i = \text{``bad''}] =  \frac{1}{2^{|V|}}
    \end{align*}
    So, we know that:
    \begin{align*}
         \mathbb{P}[(t_{i}=\text{``bad''}) \cap (i\not\in B)]   = \mathbb{P}[(i\not\in B) \; | \; (t_{i}=\text{``bad''}) ] \cdot \mathbb{P}[(t_{i}=\text{``bad''})] = \frac{ 1-\pi }{2^{|V|}}
    \end{align*}
    By taking the union bound over all agents we get:
    \begin{align*}
        \mathbb{P}[\exists i\ \text{s.t.}\ (t_{i}=\text{``bad''} \cap i\not\in B)] \leq \frac{(1-\pi)s}{2^{|V|}} 
    \end{align*}
    Upper bounding this probability by $\delta/2$ and solving for $|V|$ gives us the lower bound:
    \begin{align*}
        |V| \geq \log \left( \frac{2(1-\pi)s}{\delta} \right) 
    \end{align*}
\end{proof}

\redundancy*

\begin{proof}

Assume that \textsc{CrowdSort}($T, s, v, r$) must return all $\binom{n}{2}$ correct pairwise comparisons  for the ground-truth ordering to be retireved. 
Since $|V| \geq \log \left( \frac{2(1-\pi)s}{ \delta} \right)$ we know from Lemma \ref{ver} that with high probability \textsc{CrowdSort}($T, s, v, r$) identifies all ``bad''  agents. Therefore, it suffices to assign each comparison to at least one `good''  agent, for the true value of the comparison to be returned.  The number of ``good''  agents in a set of $r$ agents follows a binomial distribution:
\begin{align*}
    \mathbb{P}[ t_i =  \text{bad} \; \forall a_i \in [r]] &= (1-\pi)^r
\end{align*}
so by the union bound over all comparisons we have that the probability that there exists a pairwise comparison that is assinged to $r$ ``bad'' agents is at most:
\begin{align*}
     n^2\mathbb{P}[ t_i =  \text{bad} \; \forall a_i \in [r]] = n^2 (1-\pi)^r \leq \frac{\delta}{2}
\end{align*}
Thus, we get that:
\begin{align*}
    r \geq \frac{\log \left( \frac{\delta}{n^2}\right)}{\log (1-\pi)}
\end{align*}

Now let $E_{1}$ be the event that $B = \{i | \forall a_i$ s.t. $t_i =$``bad''$\}$ and  $E_2$ be the event that every pairwise comparison is assigned to at least one ``good'' agent. For \textsc{CrowdSort}($T, s, v, r$) to return the items' ground truth ordering with probability at least $1-\delta$, it suffices that all pairwise comparisons are assigned to at least one ``good'' agent and that all comparisons performed by ``bad'' agents are eliminated i.e. all ``bad'' agents are identified. This holds if $\mathbb{P}[ E_1 \cap E_2] \geq 1-\delta$. 
\par By Lemma \ref{ver} we have that $\mathbb{P}[E_{1}] \geq 1-\frac{\delta}{2}$. We also know $\mathbb{P}[E_{2}] \geq 1-\frac{\delta}{2}$. We want to show that $\mathbb{P}[E_{1} \cap E_{2}] \geq 1-\delta$. It suffices to show that $\mathbb{P}[\neg E_{1} \cup \neg E_{2}] \leq \delta$. This holds because  by the union bound $\mathbb{P}[\neg E_{1} \cup \neg E_{2}] \leq \mathbb{P}[\neg E_{1}] + \mathbb{P}[E_{2}] \leq 2\frac{\delta}{2} = \delta$.
\end{proof}

\numcomps*
\begin{proof}
In \textsc{findWset}, the algorithm \textsc{findPrime} finds the smallest prime $q$, such that $q \geq \sqrt{|T|}$. In \textsc{findWset}, each agent is assigned all pairwise comparisons among non-hallucinated items within a subset, for $r = \mathcal{O}(\log{n})$ subsets. Each subset consists of $q$ items, so each agent will be asked to compare the following number of items:
\begin{align*}
    r\cdot \frac{n}{s} \cdot q & \leq O\left(\log{n} \cdot \frac{n}{s} \cdot \left( \sqrt{|T|} + (q-\sqrt{|T|})\right) \right)
    \\ & \leq O\left(\log{n} \cdot \left( \sqrt{n} + (q-\sqrt{n})\right) \right)
    \\ & \leq O \left( \frac{n}{s} \log{n} \cdot (\sqrt{n} + (\log{\sqrt{n}})^2) \right)
    \\ & \leq \tilde{O}(n^{3/2}/s)
\end{align*}
where in the last step, we use Cramer's conjecture, which states the following:
If $p_{n}$ and $p_{n+1}$ are two consecutive prime numbers, then:
\begin{align*}
    p_{n+1} - p_{n} = \mathcal{O}((\log p_n)^2)
\end{align*}

For every subset that the agent is asked to fully sort, the agent can sort items adaptively so the number of pairwise comparisons he will perform is $O(q\log{q})$.  Therefore:

\begin{align*}
    |W_i| & \leq r \cdot \frac{n}{s} \cdot \mathcal{O}(q\log{q}) \tag{$rn/s$ subsets of $q$ items each }
    \\ & \leq \mathcal{O}\left( \frac{n}{s}\log{n} \right) \cdot \mathcal{O} \left( \left( \sqrt{|T|} + (q-\sqrt{|T|})\right) \cdot \log{\left(\sqrt{|T|} + (q-\sqrt{|T|} )\right) } \right) 
    \\ & \leq \mathcal{O}\left( \frac{n}{s}\log{n} \right) \cdot \mathcal{O} \left( \left( \sqrt{n} + (\log{\sqrt{n}})^2\right) \cdot \log{\left(\sqrt{n} + (\log{\sqrt{n}})^2 \right) } \right) 
    \\ & \leq \mathcal{O}\left( \frac{n}{s}\log{n} \right) \cdot \mathcal{O} \left( \left( \sqrt{n} + (\log{\sqrt{n}})^2\right) \cdot \log{\sqrt{n}  } \right) 
    \\ & \leq \mathcal{O}\left( \frac{n}{s}\log{n} \right) \cdot \mathcal{O} \left( \sqrt{n}\cdot \log{\sqrt{n}  } \right) 
    \\ & \leq \mathcal{O}(\frac{n^{3/2}}{s}\log^2{n})
    \\ & \leq \tilde{\mathcal{O}}(n^{3/2}/s),
\end{align*}

\end{proof}
\begin{remark}\label{rmrk}
    In expectation the number of comparisons needed to fully sort $q$ items using quicksort is at most $2q\ln{q}$. Therefore, we write that for every agent $a_i$:
\begin{align*}
    \mathbb{E}[|C_i|] & = \mathbb{E}[|V| + |W_i|] \leq |V| + 2rnq\ln{q}/s 
\end{align*}
and we set $d = |V| + 2rnq\ln{q}/s = \mathcal{O}(n^{3/2}\log^2{n}/s)$.
\end{remark}

\begin{algorithm}[t]
    \caption{SGP($T=\{T_1, \cdots T_{q^2}\}$)}\label{alg:sgp}
    \begin{algorithmic}[1]
    \State $\sigma \leftarrow$ \text{Uniform}$(S_{q^2})$
    \For{$k=1, 2, \cdots, q-1$}
        \For{$(i,j)\in \mathbb{F}_q^2$}

        \State $L_{k-1}(i,j)\leftarrow ki+j$
        \EndFor
        \For{$l\in \mathbb{F}_q$}

        \State $R_{k-1,l} \leftarrow \{\sigma^{-1}(i,j) : L_{k-1}(i,j) = l\}$
        \EndFor
        \State $P_{k-1}\leftarrow\{R_{k-1,0}, \cdots, R_{k-1,q-1} \}$ 
    \EndFor
    \For{$l \in \mathbb{F}_q$}
        \State $N_l \leftarrow \{\sigma^{-1}(0, l),\sigma^{-1}(1, l),\cdots \sigma^{-1}(q-1, l)\}$
        \State $M_l \leftarrow \{\sigma^{-1}(l, 0),\sigma^{-1}(l, 1),\cdots \sigma^{-1}(l, q-1)\}$
    \EndFor
    \State $P_{q-1}\leftarrow \{M_0, \cdots, M_{q-1}\}$
    \State $P_q\leftarrow \{N_0, \cdots, N_{q-1}\}$
\end{algorithmic}
\end{algorithm}

\begin{restatable}{theorem}{sgp}\label{thm:sgp}[\citep{SGP}]
For a prime power $q$, there exists a solution to the instance SGP($q, q, q+1$), such that every pair of golfers plays against each other exactly once.
\end{restatable}

\begin{proof}
 We want to prove that there exists a schedule such that $q^2$ players are split into disjoint groups of $q$ players each for $q$ days such that no two players are in the same group for more than one day.
 \par Let $\mathbb{F}_q$ denote the GF(q) for the prime number q. We pick an arbitrary bijection $\sigma:\mathbb{F}_q^2\rightarrow \{p_1, \cdots, p_{q^2}\}$  between $q^2$ players and ordered pairs $(i,j) \in \mathbb{F}_q^2.$ For all $1\leq k<q,$ let $L_k:\mathbb{F}_q^2\rightarrow\mathbb{F}_q$ be a function such that $L_{k}(i,j)=ki+j$.
 \par Now, it is easy to verify  that for all $i \in \{0, \cdots, q-1\}, k \in \{1, \cdots q-1\}$, it holds that $L_k(i, :)$ contains the elements $1, \cdots q-1$ exactly once. (Property 1)
 \par Also, notice that for any $i_1, i_2, j_1, j_2 \in \{0, \cdots, q-1\}$, and  $k \in \{1, \cdots, q-1\}$, if $L_k(i_1, j_1) = L_k(i_2, j_2)$, then $\nexists m$ s.t. $L_m(i_1, j_1) = L_m(i_2, j_2)$. (Property 2). To prove this, suppose that there exists $m$ s.t. $L_{m}(i_1,j_1)=L_{m}(i_2,j_2).$ Then we have $ki_1+j_1=ki_2+j_2$ and $mi_1+j_1=mi_2+j_2,$ so $(i_1-i_2)k=(i_1-i_2)m.$ Since $\mathbb{F}_q$ is a finite field and $k\neq m$, we must have $i_1=i_2,$ which implies that $j_1=j_2$, a contradiction.

 \par Now let $W_k$ be a list of q sets that denotes day i consisting of q groups. For $k \in \{0, \cdots q-2\}$, we set $W_k=(R_{k,1},\cdots, R_{k,q})$, where $R_{k,l}$ is the set of all tuples in the domain of $L_k$ that map to $l$, i.e. $R_{k,l}= \{\sigma^{-1}(i,j) : L_k(i,j) = l\}$. For day $q$, we set $W_q=\{N_0, \cdots N_{q-1}\}$, where $N_l=\{\sigma^{-1}(0, l),\sigma^{-1}(1, l),\cdots \sigma^{-1}(q-1, l)\}$ for all $0\leq l< q$. For day q-1, we set $W_q=\{M_0, \cdots M_{q-1}\}$, where $M_l=\{\sigma^{-1}(l, 0),\sigma^{-1}(l, 1),\cdots \sigma^{-1}(l, q-1)\}$ for all $0\leq l< q$. For day q, we set $W_q=\{N_0, \cdots N_{q-1}\}$, where $N_l=\{\sigma^{-1}(0, l),\sigma^{-1}(1, l),\cdots \sigma^{-1}(q-1, l)\}$ for all $0\leq l< q$.
 \par We will show that the collection $W_1, \cdots, W_q$ is a valid assignment scheme for the Social Golfer Problem. By Property 1, we know that for every day, every player is assigned in a group and that groups are disjoint. By Property 2, we know that for distinct indices $(i_1,j_1)\neq (i_2,j_2)$, if $L_k(i_1, j_1) = L_k(i_2, j_2)=l_1 $, then $\nexists m$ s.t. $L_m(i_1, j_1) = L_m(i_2, j_2)=l_2$.  Thus any two elements assigned to the same group in day $W_{k}$ cannot be assigned to the same group in day $W_{m}$ for all $1\leq k\neq m<q$. Lastly, note that for all $1\leq k < q$ and $0\leq a_1,a_2<q,$ we have $L_k(a_1, l)=ka_1+l\neq ka_2+l=L_k(a_2,l),$ so $\sigma_{-1}(a_1,l),\sigma_{-1}(a_2,l)\in N_l$ are not assigned to the same part in any of $W_1, \cdots W_{q-1},$ so partition $W_q$ satisfies the desired properties when added to the collection. A similar argument holds for $W_{q-1}$. Thus, our construction $W_1, \cdots, W_{q+1}$ is valid solution to the instance $SGP(q, q, q+1)$.

 We have proved that each player will play with a group of $q-1$ different players for $q+1$ days. There are $q^2$ players in total. Then, each player will play against a total of $(q-1)(q+1) = q^2-1$ different players. So every player will be in the same group with every other player exactly once.

\end{proof}

\bigthm*
\begin{proof}

\par Results (1), (2), (3) directly stem from Lemmas \ref{number of comparisons}, \ref{ver}, and \ref{redundancy} respectively. We denote with $u$ the utility of the principal and calculate its expected value given that she sets the payment $p=p^*$. Recall that $u_{\mu^*}$ is a constant multiple of the number of pairwise comparisons that are returned by our algorithm.
\begin{align*}
    \mathbb{E}[u | p = p^*] = \mathbb{E} \left[ u_{\mu^*} - \bar{\psi} |V| - p^* \sum_{i=1}^{s} \mathbbm{1}\{t_i = \text{``good''}\} - p^{*} \sum_{i=1}^{s}\mathbbm{1}\{t_i = \text{``bad''}, c_i = 0\}  \mid p=p^{*} \right] 
\end{align*}

By Theorem \ref{p^*} we know that if the principal sets the price to $p = p^*$ the probability that any agent is ``good'' is $\pi$. By Lemma \ref{ver} we also know that the probability that any agent is ``bad'' and not identified is $\frac{1-\pi}{2^{|V|}}$. Therefore, we can write:
\begin{align*}
    \mathbb{E}[u | p = p^*] = \mathbb{E}[u_{\mu^*}] - \bar{\psi} |V| -p^* s\pi - p^{*} \frac{s(1-\pi)}{2^{|V|}} 
\end{align*}
where here we define $|V|$ according to Lemma \ref{ver}.
Finally, by Lemma ~\ref{redundancy}, with probability at least $1-\delta$, algorithm \textsc{CrowdSort} returns all $\binom{n}{2}$ ground-truth pairwise comparisons. Therefore, $\mathbb{E}[u_{\mu^*}] \geq (1-\delta)\lambda \binom{n}{2}$ for some $\lambda \in \mathbb{R^{+}}$. So, we get that:
\begin{align*}
    \mathbb{E}[u | p = p^*] \geq \lambda(1-\delta)\binom{n}{2} - \bar{\psi} |V| -p^* s\pi - p^{*} \frac{s(1-\pi)}{2^{|V|}} 
\end{align*}
\par If the principal decides to set the price to 0, then no agents will induce effort and the principal will need to perform all pairwise comparisons on her own. The expected number of comparisons she will need to perform is $2n\log{n}$. Therefore, the disutility of the principal in this case is:
\begin{align*}
    \mathbb{E}[u | p=0] = \lambda \binom{n}{2}-2\bar{\psi}n\log{n}
\end{align*}

For the principal to decide to propose a contract, her expected utility when inducing effort from the agents must be greater than or equal to her utility when she does not. In other words, we require that:
\begin{align*}
    \mathbb{E}[u | p = p^*] \geq \mathbb{E}[u | p = 0]
\end{align*}
We substitute the values of the utilities found above along with the appropriate values for $|V| =  \log \left( \frac{2(1-\pi)s}{\delta}\right)$ to get:
\begin{align*}
    \lambda(1-\delta)\binom{n}{2}-  \bar{\psi}  \log \frac{2(1-\pi)s}{\delta}  -p^* s\pi - p^{*} \frac{s(1-\pi)}{\frac{2(1-\pi)s}{\delta}}  \geq  \lambda\binom{n}{2} -2\bar{\psi}n\log{n}
\end{align*}
which holds if and only if:
\begin{align*}
     -\lambda\delta\binom{n}{2}-p^* \left( s\pi +\frac{\delta}{2}\right) \geq  \bar{\psi}\left(\log{\frac{2(1-\pi)s}{\delta}}-2n\log{n}\right)
\end{align*}
We substitute the values $p^* = \frac{d \psi}{\hat{\pi}\pi}$ and $\hat{\pi} = 1-1/2^{|V|}$ to get:
\begin{align} \label{condition}
    \lambda\delta\binom{n}{2}+\frac{d\psi}{1-\frac{\delta}{2(1-\pi)s}} \left( s +\frac{\delta}{2\pi} \right) \leq  \bar{\psi}\left(2n\log{n} - \log{\frac{2(1-\pi)s}{\delta}}\right)
\end{align}
By Remark \ref{rmrk}, we know that $d = \mathcal{O}(n^{3/2}\log^2{n}/s)$ and assuming that $\delta < (1/n^2)$ and that $\pi$ is bounded away from $1$, we get that:
\begin{align*}
    \bar{\psi} /\psi \geq  \Omega\left(\frac{d s}{n\log{n}} \right) = \Omega(\sqrt{n}\log{n})
\end{align*}

\end{proof}

\section{Omitted proofs from Section 3}\label{appendixb}
\paymentunk*
\begin{proof}
We denote agent $a_i$'s utility as $u_i$ and following a similar analysis as in Lemma~\ref{paymentthm} we have that:
\begin{align*}
    \mathbb{E}[u_i| e_{i}=1]  = p_g\pi + p_g(1-\hat{\pi})(1-\pi) - |C_i|\psi_i
\end{align*} 
We similarly handle the second case, in which the agent does not exert effort, this time noting that the agent does not experience disutility:
\begin{align*}
    \mathbb{E}[u_i| e_{i}=0] =  p_g(1-\hat{\pi})
\end{align*} 
For agent $a_i$ to be incentivized to exert effort, the payment needs to be such that:
\begin{align*}
    \mathbb{E}[u_i|e_i =1]  \geq \mathbb{E}[u_i|e_i =0].
\end{align*} 
Using $\mathbb{E}[|C_i|] < d$ and simplifying we get the lower bound:
\begin{align} \label{eq: conndition}
    p_g \geq \frac{d\psi_i}{\hat{\pi}\pi}
\end{align}
To satisfy individual rationality, we require the following:
\begin{align*}
    &\mathbb{E}[u_i| e_{i}=1] \geq 0
\end{align*}
which holds if and only if:
\begin{align*}
p_g \geq \frac{\mathbb{E}[|C_i|]\psi_i}{1+\hat{\pi}\pi 
 - \hat{\pi}} 
\end{align*}
which is always satisfied as long as constraint \ref{eq: conndition} is satisfied.
\end{proof}

\numincentivizedagents*
\begin{proof}
Let $p_g$ be the payment function when the principal decides to incentivize $g$ agents. Then any agent exerts effort independently of all other agents if and only if $p_g \geq \frac{d\psi_i}{\hat{\pi}\pi}$. Equivalently, this happens if and only if
\begin{align*}
    \psi_i \leq \frac{p_g \hat{\pi}\pi}{d}.
\end{align*}

Thus, each agent exerts effort with probability $ q= F\left(\frac{p_g \hat{\pi}\pi}{d}\right)$ and the realized number of incentivized agents, $\hat{g}$, follows a Binomial$(s,q)$ distribution. 

We want the probability that an agent is incentivized to be at most $q \leq g/s + \epsilon$, so we require:
\begin{align}\label{2}
     F\left(\frac{p_g \hat{\pi}\pi}{d} \right) \leq \frac{g}{s} + \epsilon
\end{align}
Since we only have access to $\hat{F}$, which is increasing, it suffices that:
\begin{align*}
     \hat{F}\left(\frac{p_g \hat{\pi}\pi}{d} \right)  \leq \frac{g}{s}
\end{align*}
which holds if:
    $p_g \leq \frac{d}{\hat{\pi}\pi}\tilde{F}_{N}\left(\frac{g}{s} \right).$

\end{proof}

\versecond*
\begin{proof}
We have that for any agent $a_i$:
\begin{align*}
    &\mathbb{P}[t_i = \text{``bad''}] 
    \\&= \mathbb{P}[t_i = \text{``bad''}|e_i = 1] \cdot \mathbb{P}[e_i=1] + \mathbb{P}[e_i=0] \\
    & = (1-\pi) \cdot \mathbb{P}\left[\psi_i \leq \frac{\hat{\pi}\pi p_{g}}{d}\right] + 1-\mathbb{P}\left[\psi_i \leq \frac{\hat{\pi}\pi p_{g}}{d}\right] \\
    & = 1-\pi \cdot \mathbb{P}\left[\psi_i \leq \frac{\hat{\pi}\pi p_{g}}{d}\right].
\end{align*}

By Remark \ref{ub}, we know that for any agent $a_i$ it holds that $\mathbb{P}\left[\psi_i \leq \frac{\hat{\pi}\pi p_{g}}{d}\right] \geq \frac{g}{s}-\epsilon$, which means that:
\begin{align} \label{eq9}
    \mathbb{P}[t_i = \text{``bad''}] &\leq  1-\pi \frac{g}{s} +\pi \epsilon
\end{align}
We also know that $\mathbb{P}[i \not\in B | t_i = \text{``bad''}] = 1/2^{|V|}$, so by taking the union bound over all agents, we have:
    \begin{align} \label{3}
        \mathbb{P}[\exists i\ \text{s.t.}\ (t_{i}=\text{``bad''} \cap i\not\in B)] &\leq \frac{s}{2^{|V|}} \cdot \left(1- \pi \frac{g}{s} +\pi \epsilon\right)
    \end{align}
    We then upper bound this probability by $\delta/2$ and solve for $|V|,$ giving us the following lower bound:
\begin{align}\label{eq10}
        |V| \geq \log_2 \left( \frac{2(s-\pi g +\pi s \epsilon)}{\delta} \right).
\end{align}
\end{proof}

\redundancysecond*
\begin{proof}
We know from Lemma \ref{ver2} that if $|V| \geq \log_2 \left( \frac{2(s-\pi g +\pi s \epsilon)}{\delta} \right)$, then with high probability, the algorithm $\textsc{CrowdSort}(T, s, v, r)$ identifies all ``bad''  agents. Therefore, it suffices to assign each comparison to at least one ``good''  agent for the true value of the comparison to be returned. By Remark \ref{ub}, we know that for any agent $a_i$, $\mathbb{P}_{\psi_i \sim \mathcal{D}}\left[\psi_i \leq \frac{\hat{\pi}\pi p_{g}}{d}\right] \geq \frac{g}{s}-\epsilon.$ So the probability that all $r$ agents assigned to a single comparison are ``bad'' is at most: 
\begin{align} \label{4}
    \mathbb{P}( t_i =  \text{bad} \; \forall a_i \in [r]) &\leq \left(1-\frac{\pi g}{s} +\pi \epsilon \right)^r.
\end{align}
Taking the union bound over all comparisons, we have that the probability there exists a comparison such that all agents assigned to it are ``bad'' is at most $n^2 \left(1-\frac{\pi g}{s} +\pi \epsilon\right)^r.$
Upper-bounding the above probability by $\delta/2$,  we get that:
\begin{align} \label{r}
    r \geq \frac{\log_2 \left( \frac{\delta}{2n^2}\right)}{\log_2 \left( 1-\frac{\pi g}{s} +\pi \epsilon \right)}
\end{align}
Finally, notice that the above inequality and Inequality~\eqref{3} hold simultaneously with probability at least $1-\delta/2$.

Now let $E_{1}$ be the event that $B = \{i | \forall a_i$ s.t. $t_i =$``bad''$\}$ and  $E_2$ be the event that every pairwise comparison is assigned to at least one ``good'' agent. For \textsc{CrowdSort}($T, s, v, r$) to return the items' ground truth ordering with probability at least $1-\delta$, it suffices that all pairwise comparisons are assigned to at least one ``good'' agent and that all comparisons performed by ``bad'' agents are eliminated i.e. all ``bad'' agents are identified. This holds if $\mathbb{P}[ E_1 \cap E_2] \geq 1-\delta$. 
\par By Lemma \ref{ver2} we have that $\mathbb{P}[E_{1}] \geq 1-\frac{\delta}{2}$. We also know that if \ref{r} holds, then $\mathbb{P}[E_{2}] \geq 1-\frac{\delta}{2}$. We want to show that $\mathbb{P}[E_{1} \cap E_{2}] \geq 1-\delta$. It suffices to show that $\mathbb{P}[\neg E_{1} \cup \neg E_{2}] \leq \delta$. This holds because  by the union bound $\mathbb{P}[\neg E_{1} \cup \neg E_{2}] \leq \mathbb{P}[\neg E_{1}] + \mathbb{P}[E_{2}] \leq 2\frac{\delta}{2} = \delta$.
\end{proof}

\optg*

\begin{proof} 
    Results (1), (2), (3) directly stem from Lemmas \ref{number of comparisons}, \ref{ver2}, and \ref{redundancy2} respectively. We denote as $v$ the utility of the principal and calculate its expected value given that she decides to incentivize at least $g$ agents and set the payment $p=p_g$. Recall that $u_{\mu^*}$ is the utility of the principal if the resulting sorting of the items coincides with the ground-truth permutation $\mu^*$. Here, $u_\mu^*$ scales with the number of comparisons the principal retrieves from the agents.
\begin{align*}
    \mathbb{E}[u | p = p_{g}] = \mathbb{E} \left[ u_{\mu^*} - \bar{\psi} |V| - p_g \sum_{i=1}^{s} \mathbbm{1}\{t_i = \text{``good''}\}  - p_g\sum_{i=1}^{s}\mathbbm{1}\{t_i = \text{``bad''}, c_i = 0\}  \mid p=p_g \right]
\end{align*}
By Remark \ref{ub} we know that if the principal sets the price to $p = p_{g}$ then with probability at least $1-\delta/4$ the probability that any given agent is ``good'' is at most $\pi(\frac{g}{s}+\epsilon)$. By Lemma \ref{ver2} we also know that with probability at least $1-\delta/4$ the probability that any agent is ``bad'' and not identified is at most $\frac{1}{2^{|V|}}(1-\frac{\pi g}{s}+\pi \epsilon)$. We substitute $u_{\mu^*} \geq \lambda(1-\delta)\binom{n}{2}$ for some $\lambda \in  \mathbb{R}^{+}$ by Lemma \ref{redundancy2} . Therefore, we have that:
\begin{align*}
    & \mathbb{E}[u | p = p_{g}] \geq \lambda(1-\delta)\binom{n}{2} - \bar{\psi} |V| -sp_{g}\left(\left(1-\frac{\delta}{4}\right)\pi\left(\frac{g}{s}+\epsilon\right) +\frac{\delta}{4} \right) -  sp_g \left(\left(1-\frac{\delta}{4}\right)\frac{1}{2^{|V|}}\left(1-\pi \frac{g}{s}+\pi \epsilon\right)+\frac{\delta}{4}\right)
\end{align*}

where here we take $|V|$ satisfying the conditions from Lemma \ref{ver2} and using the fact that $\hat{\pi} = \mathbb{P}[c_{i}=1 | t_{i}=\text{``bad''}] = 1-1/2^{|V|}$.

Let $E(p_g)$ denote the lower bound above. For the algorithm to be incentive-compatible for the principal, it must hold that:
\begin{align*}
    \mathbb{E}[u | p = p_{g}] \geq E(p_g)\geq\mathbb{E}[v | p = 0] = u_{\mu^*} - 2\bar{\psi}n\log{n}
\end{align*}
where the first inequality follows from above. Therefore the optimization problem for a suitable $g^*$ is the following:
\begin{align*}
    g^* &= \arg \max_{g\in[s]} \; E(p_g) 
\end{align*}

such that:

\begin{align*}
    \left.\begin{aligned}
 |V| &\geq \log_2 \left( \frac{2(s-\pi g +\pi s \epsilon)}{\delta} \right)
\\  r_g &\geq \frac{\log_2 \left( \frac{\delta}{2n^2}\right)}{\log_2 \left( 1-\frac{\pi g}{s} +\pi \epsilon\right)}
    \end{aligned}\right\} \text{correctness of algorithm}
\end{align*}

\begin{align*}
\left.\begin{aligned}
    p_g = \frac{d}{\hat{\pi}\pi}\tilde{F}_{N}\left(\frac{g}{s}\right)
\end{aligned}\right\} \text{relaxed incentive compatibility for agents}
\end{align*}

\begin{align*}
\left.\begin{aligned}
E(p_g) \geq u_{\mu^*} - 2\bar{\psi}n\log{n}
\end{aligned}\right\} \text{incentive compatibility for principal}
\end{align*}

Notice that the objective is increasing in $|V|$ and in $p_g$ (therefore also in $d$). Thus, the inequality constraints above can be turned into equality constraints, fixing the values of $|V|$ and $r$ to their minimum feasible values without changing the outcome of the optimization. Therefore, the above problem can be solved in $O(s)$ time by iterating over all $g \in [s]$.

\end{proof}
\section{Additional Information about Experiments} \label{appendixexperiments}

\subsection{Post-processing algorithm}

\begin{algorithm}[H] 
    \caption{Post-Process(comparisons, identified)} \label{postprocess}
    \begin{algorithmic}[1] 
    \State condensed = \{\}
    \State conflicts = \{\}
    \For{agent $\in$ s}
        \If{agent not in identified}
            \For{``$i < j$" in comparisons[agent]}
                \If{``$i < j$" in condensed or ``$i < j$" in conflicts}
                    \State continue
                \EndIf
                \If {``$j < i$" in condensed}
                    \State conflicts $\leftarrow$ conflicts $\cup$ \{``$j < i$", ``$i < j$"\}
                    \State  condensed $\leftarrow$ condensed $\setminus$ \{``$i < j$"\}
                \Else
                    \State condensed $\leftarrow$ condensed $\cup$ \{``$i < j$"\}
                \EndIf
            \EndFor
        \EndIf
    \EndFor
    \State Return condensed
\end{algorithmic}
\end{algorithm}

\par Algorithm \ref{postprocess} takes as input ``comparisons'', which is a list of $s$ sublists, each containing the returned results of the binary comparisons performed by an agent. Each result is represented with a string of the form ``$i < j$'', representing that according to the agent, the true score of item $T_i$ is smaller than that of item $T_j$. The input ``identified'' is a set containing the indeces of all agents who returned at least one verified comparison incorrectly. 
\par Algorithm \ref{postprocess} iterates through the comparisons returned by agents not in the set ``identified'' and adds them to the set ``condensed'', which will be returned at the end. If the comparison is already in the set ``condensed'', then it is not added again. If there has been another agent whose returned comparison disagrees with it, then the comparison is added to the set ``conflicts'' (if not already added), and is removed from the set ``condensed''. Consequently, if there is at least one agent that disagrees on a comparison with all other agents assigned the same comparison, then algorithm \ref{postprocess} will completely disregard it. Finally, algorithm \ref{postprocess} returns the set condensed, which has size at most  $\binom{n}{2}$.

\subsection{Principal's utility vs. parameter $\pi$}
The plot shows a monotonic increase in the principal's utility as $\pi$ increases. The shaded region corresponds to a 90\% confidence interval across 15 iterations. We notice that for any value of $\pi$ that is greater than $0.35$, it is in the principal's best interest to implement our contract instead of performing all pairwise comparisons herself.

\begin{figure}[h]
   \centering
  \includegraphics[width=0.4\textwidth]{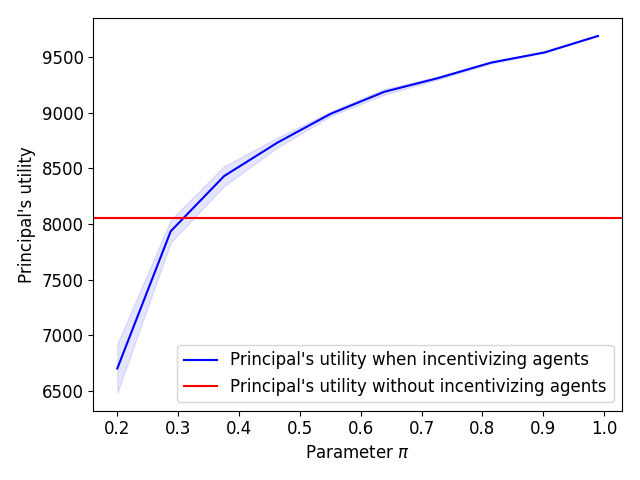}
   \caption{Principal's utility as a function of parameter $\pi$. 
    }
   \label{fig:utilityvspi}
\end{figure}

\subsection{Additional Hyper-parameter settings from Section \ref{sec:experiments}}
For the remainder of this section, we fix the hyper-parameter settings found in Section \ref{sec:experiments}, unless otherwise stated. Namely, we assume that $\pi = 0.8$, $\delta = 0.01$, $n=100$, $s = 100$, $\bar{\psi} = 2$, and $ \psi = 0.01$.

\smallskip
\emph{Discussion of Figure \ref{extrautilityvspsi1}.} 
In Figure \ref{extrautilityvspsi1} we provide additional hyper-parameter settings for Figure \ref{fig:utilityvspsi} in Section \ref{sec:experiments}. First, we select the probability of an effortful agent being ``good'' to be a value in $\{0.3, 0.6, 0.9\}$ and the number of agents and items to be in the set $\{50, 200\}$. In plots \ref{fig:utilityvspsi0.350}-\ref{fig:utilityvspsi0.950} results are averaged across 15 trials, while in plots \ref{fig:utilityvspsi0.3200}- \ref{fig:utilityvspsi0.9200}, results are averaged across 10 trials. The blue shading corresponds to a 90\% confidence interval. We notice that the confidence intervals decrease with the number of agents, even for a small number of trials. The confidence intervals also decrease, as the agents' work becomes more predictable ($\pi$ increases).

We find that as $\pi$ increases, the principal's utility also increases along with the threshold for $\psi$, below which it is in the principal's best interest to run \textsc{CrowdSort} than perform all comparisons herself. This threshold also decreases with the number of agents, $s$. This implies that there is a sweet spot, in which the expected number of ``good'' agents is just enough that the ground-truth ordering of the items is retrieved and it is sufficiently small so that the principal can pay all of them.
\begin{figure}[h]
\centering
\begin{subfigure}{0.32\textwidth}
    \includegraphics[width=\textwidth]{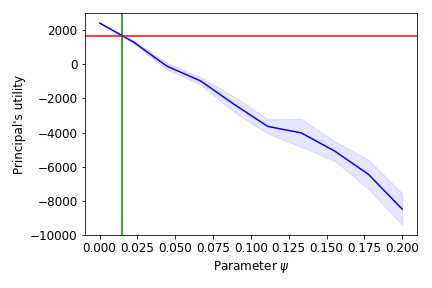}
    \caption{$\pi = 0.3$ and $s = 50$}
    \label{fig:utilityvspsi0.350}
\end{subfigure}
\hfill
\begin{subfigure}{0.32\textwidth}
    \includegraphics[width=\textwidth]{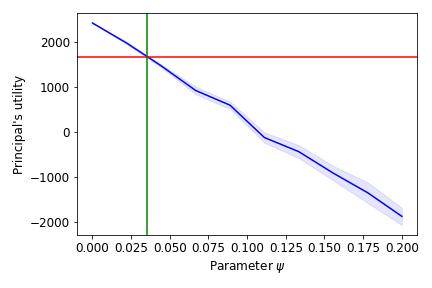}
    \caption{$\pi = 0.6$ and $s = 50$}
    \label{fig:utilityvspsi0.650}
\end{subfigure}
\hfill
\begin{subfigure}{0.32\textwidth}
    \includegraphics[width=\textwidth]{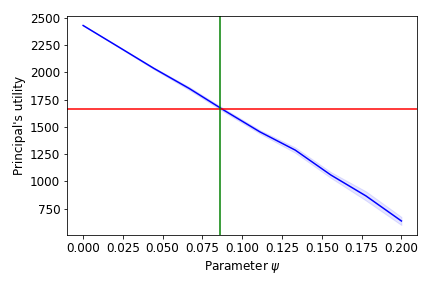}
    \caption{$\pi = 0.9$ and $s = 50$}
    \label{fig:utilityvspsi0.950}
\end{subfigure}
\hfill
\begin{subfigure}{0.32\textwidth}
    \includegraphics[width=\textwidth]{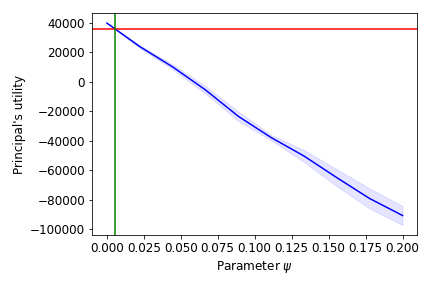}
    \caption{$\pi = 0.3$ and $s = 200$}
    \label{fig:utilityvspsi0.3200}
\end{subfigure}
\hfill
\begin{subfigure}{0.32\textwidth}
    \includegraphics[width=\textwidth]{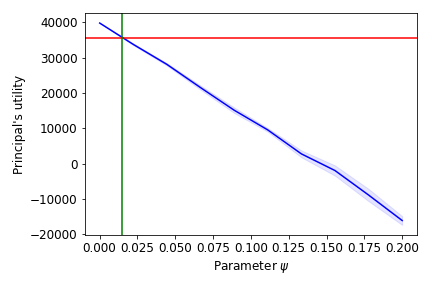}
    \caption{$\pi = 0.6$ and $s = 200$}
    \label{fig:utilityvspsi0.6200}
\end{subfigure}
\hfill
\begin{subfigure}{0.32\textwidth}
    \includegraphics[width=\textwidth]{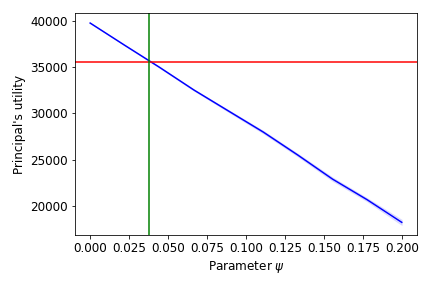}
    \caption{$\pi = 0.9$ and $s = 200$}
    \label{fig:utilityvspsi0.9200}
\end{subfigure}
\caption{Principal's utility as a function of parameter $\psi$ (in the absence of noise)}
\label{extrautilityvspsi1}
\end{figure}

\smallskip
\emph{Discussion of Figure \ref{extrautilityvspsi2}.} 
In Figure \ref{extrautilityvspsi2} we provide additional hyper-parameter settings for Figure \ref{fig:utilityvspsi} in Section \ref{sec:experiments}. Specifically, we select the per-comparison disutility of the principal, $\bar{\psi}$, along with the per-comparison utility of the principal, $\lambda$, to take values in $\{4, 6, 8\}$. The number of agents and items is in the set $\{200, 500\}$ and our results are averaged across 2 trials with the blue shading corresponding to a 90\% confidence interval.

We notice that the utility of the principal dramatically increases with the value she assigns to the correct ordering. Interestingly, by comparing Figure \ref{fig:utilityvspsi0.9200} to Figures \ref{fig:utilityvspsi4200}-\ref{fig:utilityvspsi4200}, we find that even when the reliability of the agents is smaller ($\pi = 0.8$ in \ref{fig:utilityvspsi0.9200} and $\pi = 0.9$ in \ref{fig:utilityvspsi4200}-\ref{fig:utilityvspsi8200}), the principal's point of indifference is higher when her disutility is larger. As before, we notice that when a large number of agents is incentivized, their labor needs to be sufficiently inexpensive (smaller threshold on $\psi$) for their work to be more profitable to the principal than performing all pairwise comparisons herself.
Lastly, our results show, that for a broad range of hyperparameters, our theoretical analysis accurately predicts the value of $\psi$, for which the principal is indifferent between running \textsc{CrowdSort} or performing all comparisons herself.

\begin{figure}[h]
\centering
\begin{subfigure}{0.32\textwidth}
    \includegraphics[width=\textwidth]{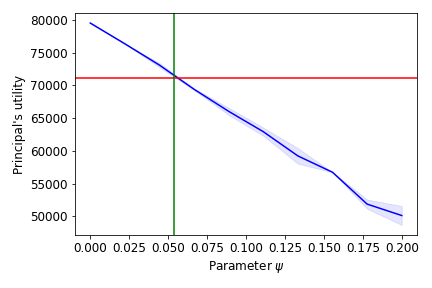}
    \caption{$\bar{\psi} = 4$ and $s = 200$}
    \label{fig:utilityvspsi4200}
\end{subfigure}
\hfill
\begin{subfigure}{0.32\textwidth}
    \includegraphics[width=\textwidth]{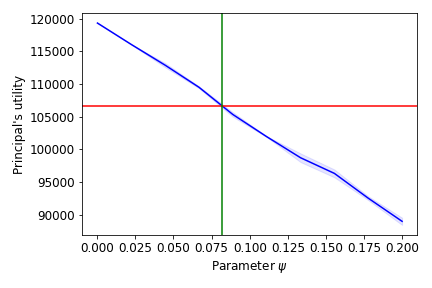}
    \caption{$\bar{\psi} = 6$ and $s = 200$}
    \label{fig:utilityvspsi6200}
\end{subfigure}
\hfill
\begin{subfigure}{0.32\textwidth}
    \includegraphics[width=\textwidth]{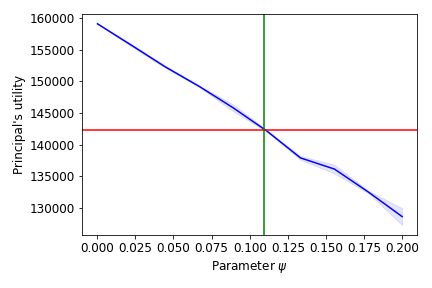}
    \caption{$\bar{\psi} = 8$ and $s = 200$}
    \label{fig:utilityvspsi8200}
\end{subfigure}
\hfill
\begin{subfigure}{0.32\textwidth}
    \includegraphics[width=\textwidth]{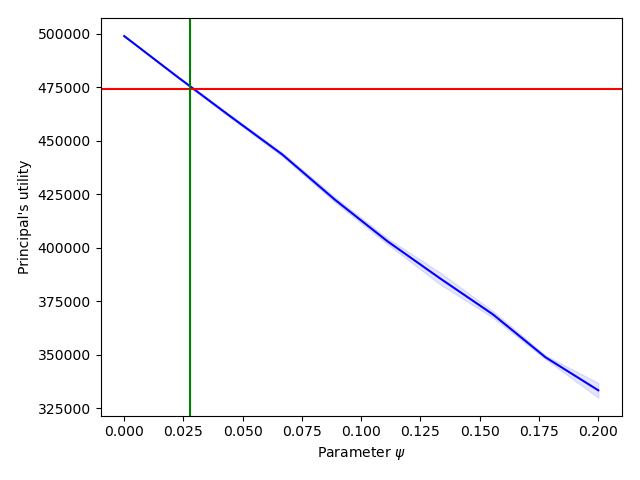}
    \caption{$\bar{\psi} = 4$ and $s = 500$}
    \label{fig:utilityvspsi4500}
\end{subfigure}
\hfill
\begin{subfigure}{0.32\textwidth}
    \includegraphics[width=\textwidth]{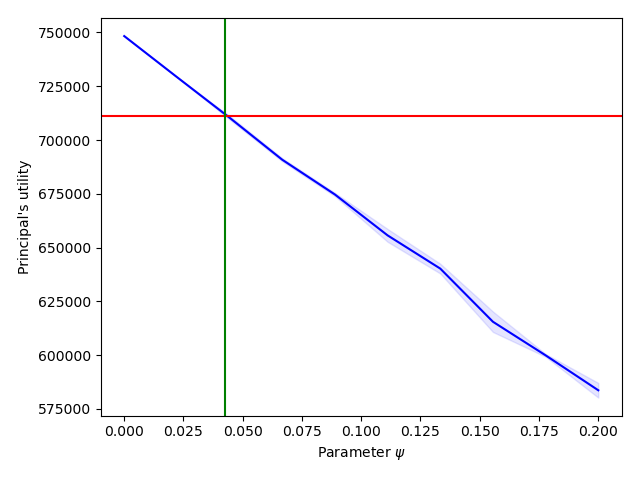}
    \caption{$\bar{\psi} = 6$ and $s = 500$}
    \label{fig:utilityvspsi6500}
\end{subfigure}
\hfill
\begin{subfigure}{0.32\textwidth}
    \includegraphics[width=\textwidth]{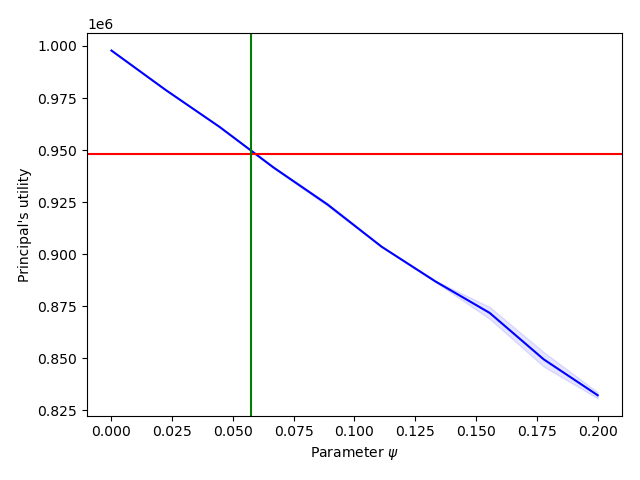}
    \caption{$\bar{\psi} = 8$ and $s = 500$}
    \label{fig:utilityvspsi8500}
\end{subfigure}
\caption{Principal's Utility as a function of parameter $\psi$ (in the absence of noise)}
\label{extrautilityvspsi2}
\end{figure}

\begin{figure}[h]
\centering
\begin{subfigure}{0.32\textwidth}
    \includegraphics[width=\textwidth]{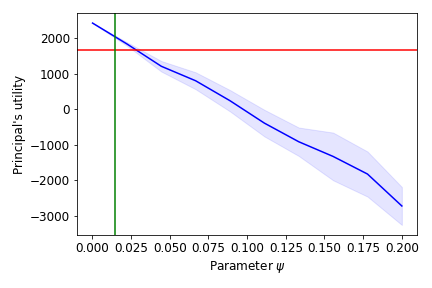}
    \caption{$\pi = 0.3$ and $s = 50$}
    \label{fig:approxutilityvspsi0.350}
\end{subfigure}
\hfill
\begin{subfigure}{0.32\textwidth}
    \includegraphics[width=\textwidth]{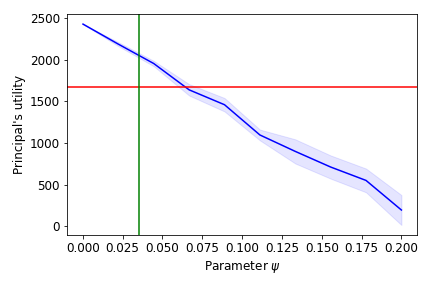}
    \caption{$\pi = 0.6$ and $s = 50$}
    \label{fig:approxutilityvspsi0.650}
\end{subfigure}
\hfill
\begin{subfigure}{0.32\textwidth}
    \includegraphics[width=\textwidth]{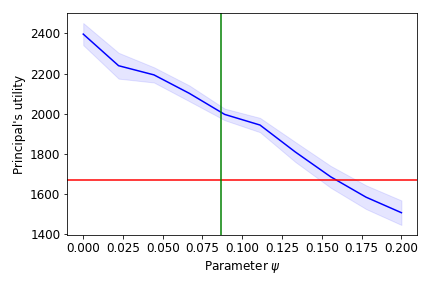}
    \caption{$\pi = 0.9$ and $s = 50$}
    \label{fig:approxutilityvspsi0.950}
\end{subfigure}
\hfill
\begin{subfigure}{0.32\textwidth}
    \includegraphics[width=\textwidth]{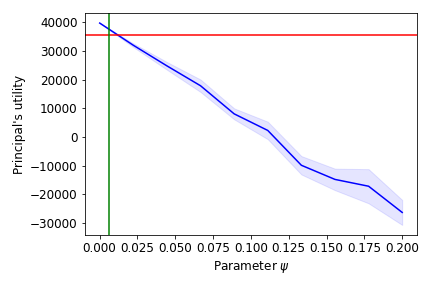}
    \caption{$\pi = 0.3$ and $s = 200$}
    \label{fig:approxutilityvspsi0.3200}
\end{subfigure}
\hfill
\begin{subfigure}{0.32\textwidth}
    \includegraphics[width=\textwidth]{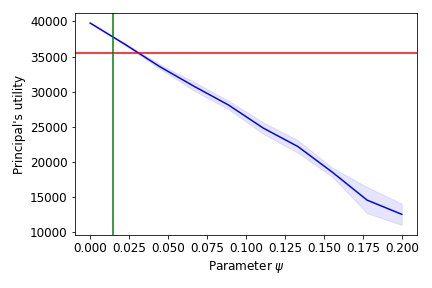}
    \caption{$\pi = 0.6$ and $s = 200$}
    \label{fig:approxutilityvspsi0.6200}
\end{subfigure}
\hfill
\begin{subfigure}{0.32\textwidth}
    \includegraphics[width=\textwidth]{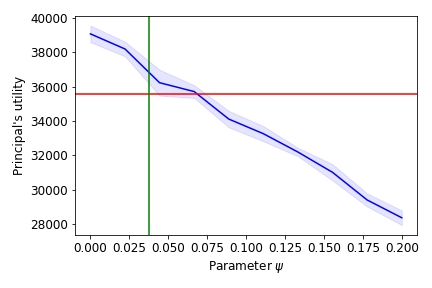}
    \caption{$\pi = 0.9$ and $s = 200$}
    \label{fig:approxutilityvspsi0.9200}
\end{subfigure}
\caption{Principal's utility as a function of parameter $\psi$ (with noise)}
\label{fig:extraapproxutilityvspsi}
\end{figure}

\smallskip
\emph{Discussion of Figure \ref{fig:extraapproxutilityvspsi}.}
We provide additional parameter settings corresponding to Figure \ref{fig:approxutilityvspi}. The parameters $\pi$ and $s$ along with the number of trials vary as in Figure \ref{extrautilityvspsi1}. 
We find that as the number of agents to be incentivized grows, our theoretical analysis for the principal's point of indifference becomes more accurate. Also, noise on the disutilites results in many agents not being incentivized to exert effort and therefore get paid. This increases the principal's utility, especially when parameter $\pi$ is large, therefore resulting in a less steep slope. As in Figure \ref{fig:approxutilityvspi}, we note that the confidence intervals become larger with $\psi$, as the principal's utility shows higher variance, when less agents are ``good''.

\smallskip
\emph{Discussion of Figures \ref{fig:extracompsvspsi} and \ref{fig:extracompsvspi}.}
We provide additional parameter settings corresponding to Figure \ref{fig:comps}. In Figure \ref{fig:extracompsvspsi}, we select the probability of an effortful agent being ``good'' to be a value in $\{0.3, 0.6, 0.9\}$, while in Figure \ref{fig:extracompsvspi}, we select the disutility of an effortful agent to be a value in $\{0.02, 0.04, 0.06\}$. The number of agents $s$ is in the set $\{50, 200\}$. In plots \ref{fig:compsvspsi0.350}-\ref{fig:compsvspsi0.950}  and 9a-9c results are averaged across 10 trials, while in plots \ref{fig:compsvspsi0.3200}-\ref{fig:compsvspsi0.9200} and \ref{fig:compsvspi0.02200}-\ref{fig:compsvspi0.06200}, results are averaged across 5 trials. 
In Figure \ref{fig:extracompsvspsi}, we see that the payment required to incentivize agents increases with the number of items, which amplifies the effect of noisy disutilites on the number of comparisons returned. Also, we find that more reliable agents (higher values of $\pi$) significantly reduce the effect of noise on the number of comparisons returned as they converge faster to the optimal $\binom{n}{2}$. Comparing Figures \ref{fig:extracompsvspsi} and \ref{fig:extracompsvspi}, we find that the noisy values of $\pi$ have a smaller impact on the number of comparisons returned compared to noisy values of $\psi$.
  
\begin{figure}[h]
\centering
\begin{subfigure}{0.32\textwidth}
    \includegraphics[width=\textwidth]{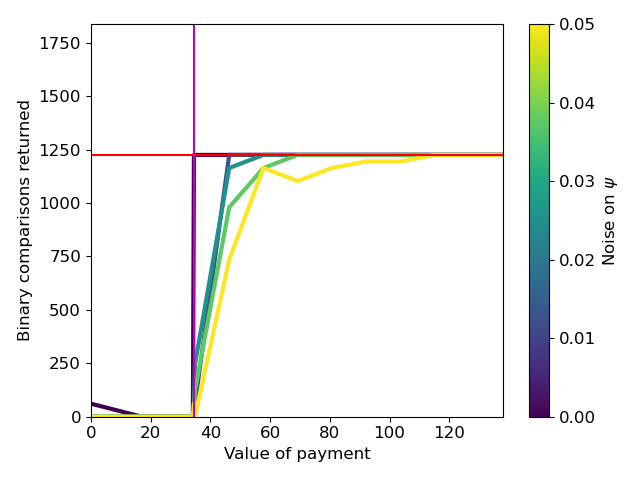}
    \caption{$\pi = 0.3$ and $s = 50$}
    \label{fig:compsvspsi0.350}
\end{subfigure}
\hfill
\begin{subfigure}{0.32\textwidth}
    \includegraphics[width=\textwidth]{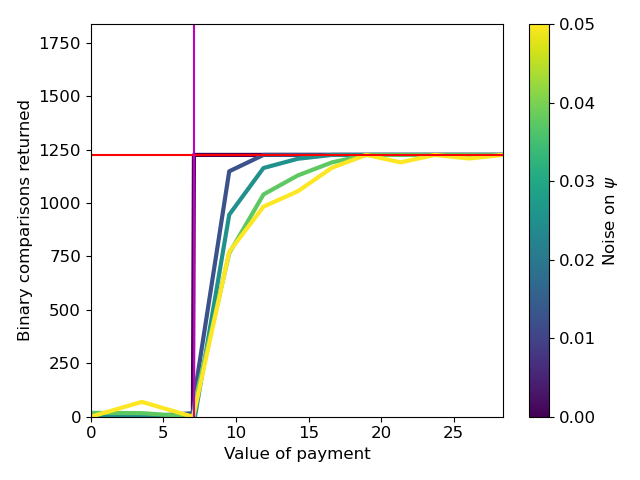}
    \caption{$\pi = 0.6$ and $s = 50$}
    \label{fig:compsvspsi0.650}
\end{subfigure}
\hfill
\begin{subfigure}{0.32\textwidth}
    \includegraphics[width=\textwidth]{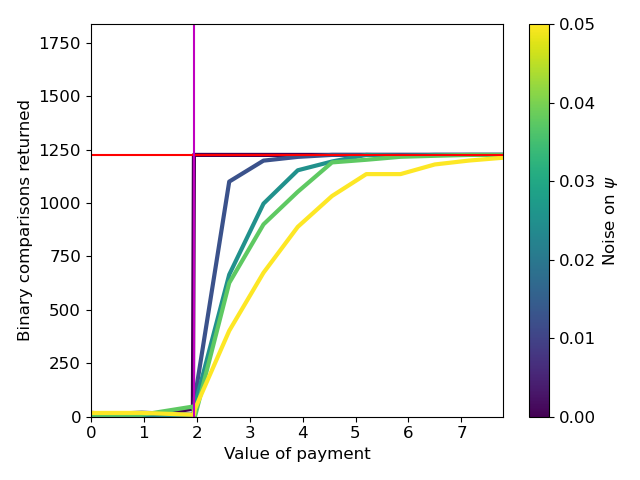}
    \caption{$\pi = 0.9$ and $s = 50$}
    \label{fig:compsvspsi0.950}
\end{subfigure}
\hfill
\begin{subfigure}{0.32\textwidth}
    \includegraphics[width=\textwidth]{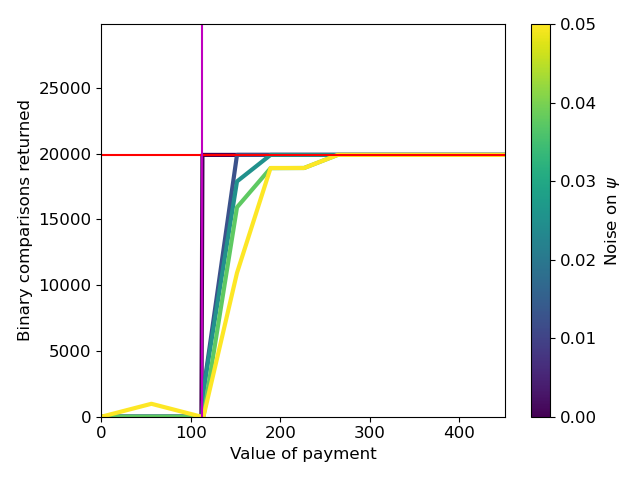}
    \caption{$\pi = 0.3$ and $s = 200$}
    \label{fig:compsvspsi0.3200}
\end{subfigure}
\hfill
\begin{subfigure}{0.32\textwidth}
    \includegraphics[width=\textwidth]{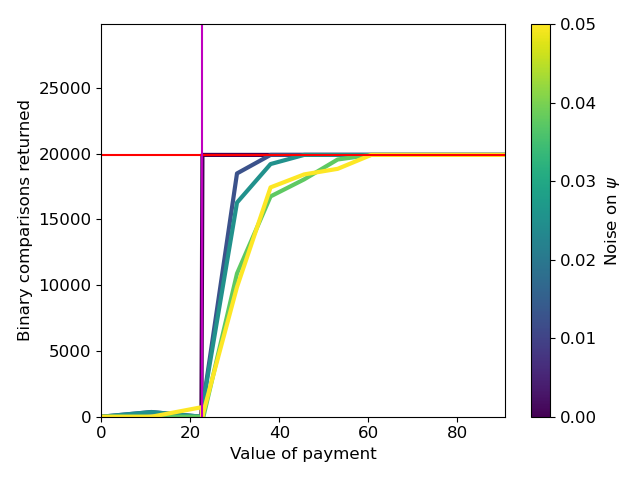}
    \caption{$\pi = 0.6$ and $s = 200$}
    \label{fig:compsvspsi0.6200}
\end{subfigure}
\hfill
\begin{subfigure}{0.32\textwidth}
    \includegraphics[width=\textwidth]{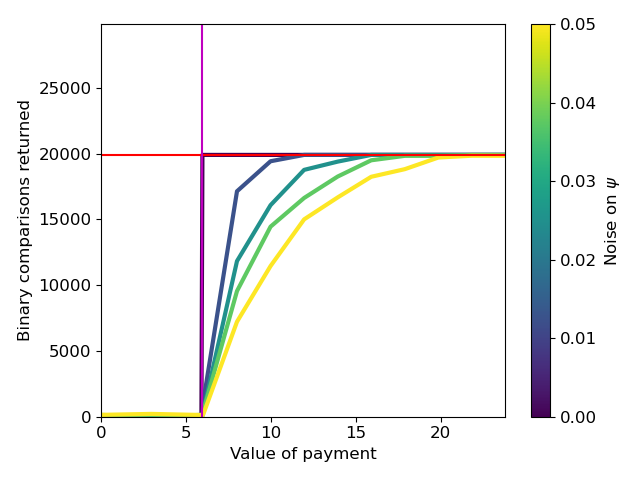}
    \caption{$\pi = 0.9$ and $s = 200$}
    \label{fig:compsvspsi0.9200}
\end{subfigure}
        
\caption{Number of comparisons returned as a function of parameter $\psi$}
\label{fig:extracompsvspsi}
\end{figure}

\begin{figure}[h]
\centering
\begin{subfigure}{0.32\textwidth}
    \includegraphics[width=\textwidth]{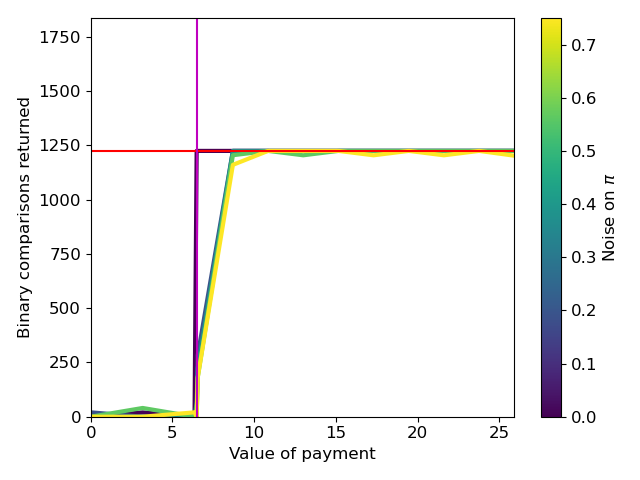}
    \caption{$\psi = 0.02$ and $s = 50$}
    \label{fig:compsvspi0.0250}
\end{subfigure}
\hfill
\begin{subfigure}{0.32\textwidth}
    \includegraphics[width=\textwidth]{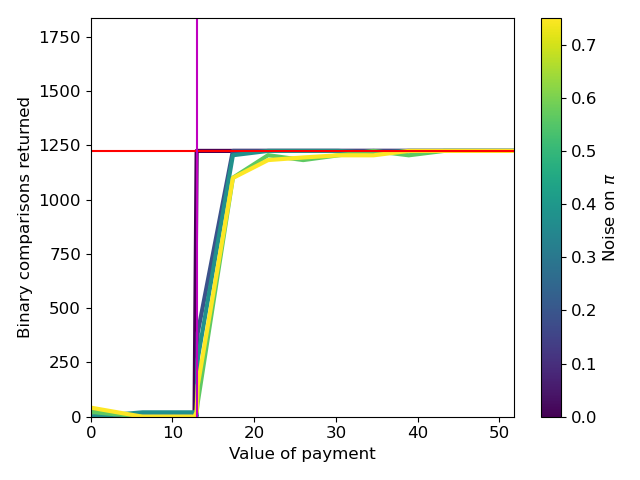}
    \caption{$\psi = 0.04$ and $s = 50$}
    \label{fig:compsvspi0.0450}
\end{subfigure}
\hfill
\begin{subfigure}{0.32\textwidth}
    \includegraphics[width=\textwidth]{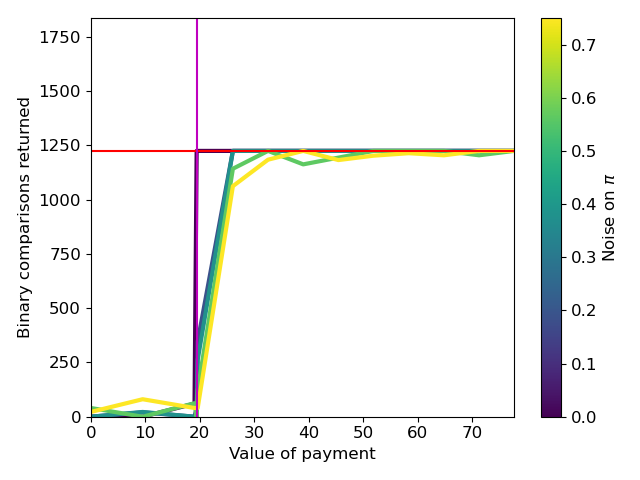}
    \caption{$\psi = 0.06$ and $s = 50$}
    \label{fig:compsvspi0.0650}
\end{subfigure}
\hfill
\begin{subfigure}{0.32\textwidth}
    \includegraphics[width=\textwidth]{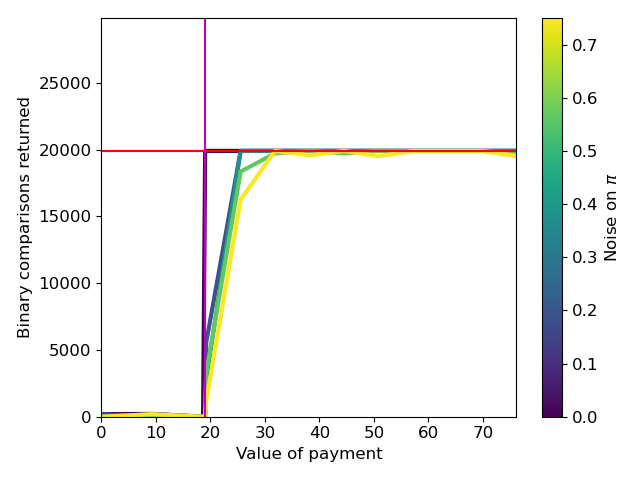}
    \caption{$\psi = 0.02$ and $s = 200$}
    \label{fig:compsvspi0.02200}
\end{subfigure}
\hfill
\begin{subfigure}{0.32\textwidth}
    \includegraphics[width=\textwidth]{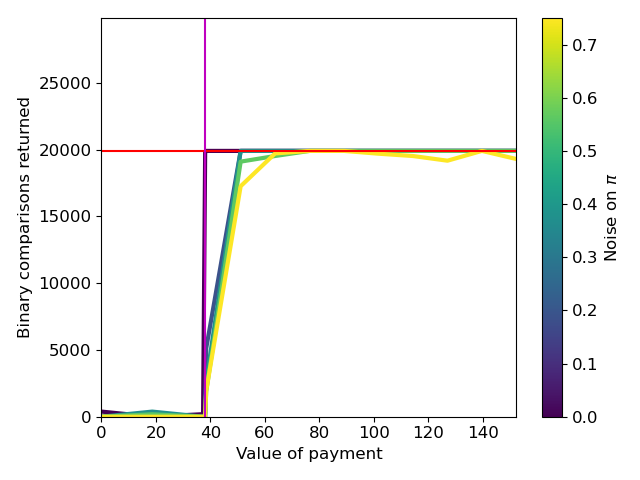}
    \caption{$\psi = 0.04$ and $s = 200$}
    \label{fig:compsvspi0.04200}
\end{subfigure}
\hfill
\begin{subfigure}{0.32\textwidth}
    \includegraphics[width=\textwidth]{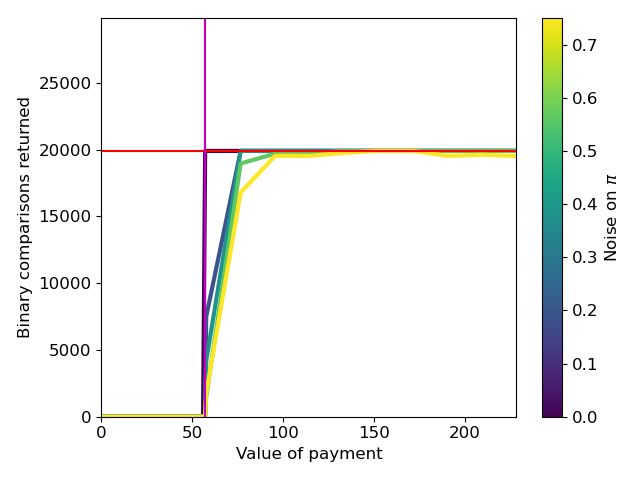}
    \caption{$\psi = 0.06$ and $s = 200$}
    \label{fig:compsvspi0.06200}
\end{subfigure}
\caption{Number of comparisons returned as a function of parameter $\pi$}
\label{fig:extracompsvspi}
\end{figure}

\smallskip
\emph{Discussion of Figure \ref{extraopt}.}
In Figure \ref{extraopt} we offer a range of parameters corresponding to Figure \ref{fig:opt} in Section \ref{sec:experiments}. Specifically, we vary the probability than an effortful agent is  ``good'', which take values in $\{0.3, 0.6, 0.9\}$ and the principal's per-comparison disutility, $\psi$, and utility $\lambda$, which takes values in $\{4, 6, 8\}$.  The multicolored plots indicate the principal's utility for different sample sizes (namely, $10, 20, 100$, and $500$) drawn from $\mathcal{D}$, which in this case is $\mathcal{N}(0.03, 0.01)$. Notice that we do not normalize the y-axis in this case. Results are averaged across 10 trials.

We see that as agents become more reliable (higher value of $\pi$) and as the principal's value for the correct sorting grows ($\lambda$), it becomes increasingly beneficial to the principal to run \textsc{CrowdSort}, than perform all pairwise comparisons on her own. Interestingly, even if $\pi$ is $0.3$, for sufficiently large values of $\bar{\psi}$, the agent's work is more profitable to the principal than sorting herself. This suggests that for a  sufficiently large value of $\bar{\psi}/{\psi}$, incentivizing agents under our proposed contract is substantially more beneficial to the principal than sorting the items herself.

\begin{figure}[h]
\centering
\begin{subfigure}{0.32\textwidth}
    \includegraphics[width=\textwidth]{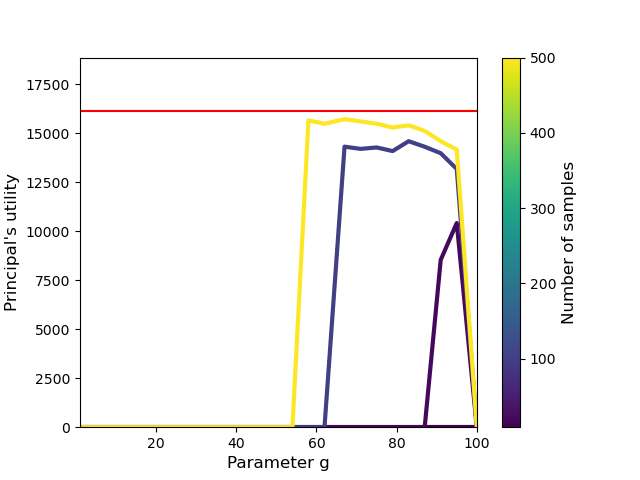}
    \caption{$\bar{\psi} = 4$ and $\pi= 0.3$}
    \label{fig:opt1000.34}
\end{subfigure}
\hfill
\begin{subfigure}{0.32\textwidth}
    \includegraphics[width=\textwidth]{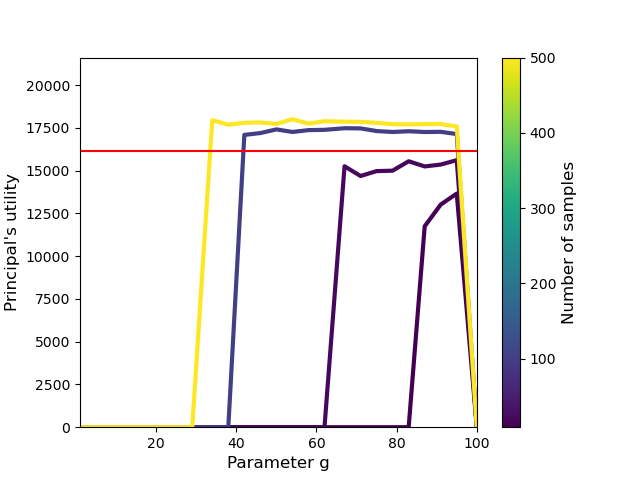}
    \caption{$\bar{\psi} = 4$ and $\pi= 0.6$}
    \label{fig:opt1000.64}
\end{subfigure}
\hfill
\begin{subfigure}{0.32\textwidth}
    \includegraphics[width=\textwidth]{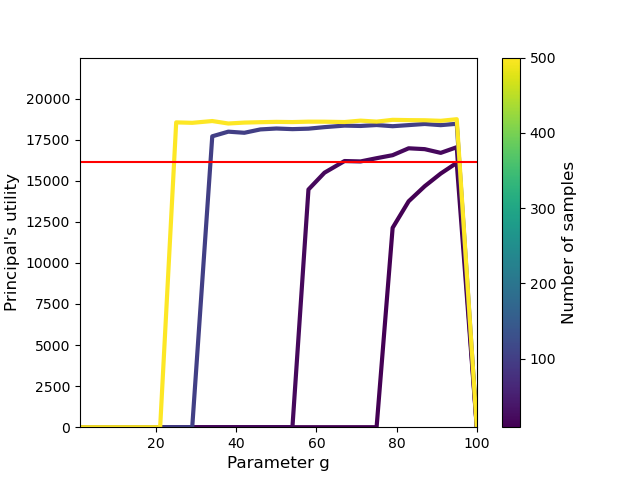}
    \caption{$\bar{\psi} = 4$ and $\pi= 0.9$}
    \label{fig:opt1000.94}
\end{subfigure}
\hfill
\begin{subfigure}{0.32\textwidth}
    \includegraphics[width=\textwidth]{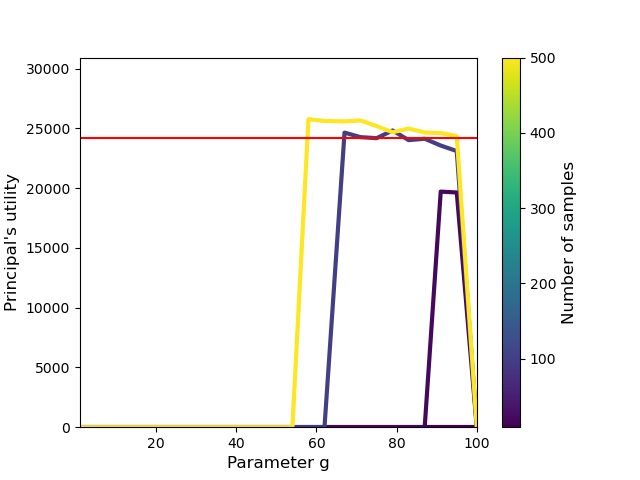}
    \caption{$\bar{\psi} = 6$ and $\pi= 0.3$}
    \label{fig:opt1000.36}
\end{subfigure}
\hfill
\begin{subfigure}{0.32\textwidth}
    \includegraphics[width=\textwidth]{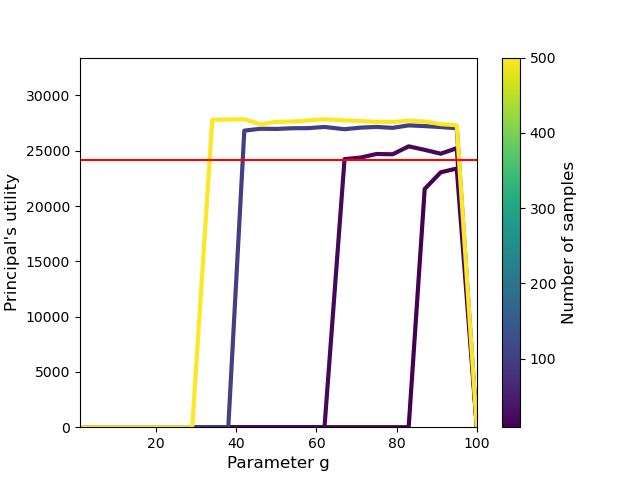}
    \caption{$\bar{\psi} = 6$ and $\pi= 0.6$}
    \label{fig:opt1000.66}
\end{subfigure}
\hfill
\begin{subfigure}{0.32\textwidth}
    \includegraphics[width=\textwidth]{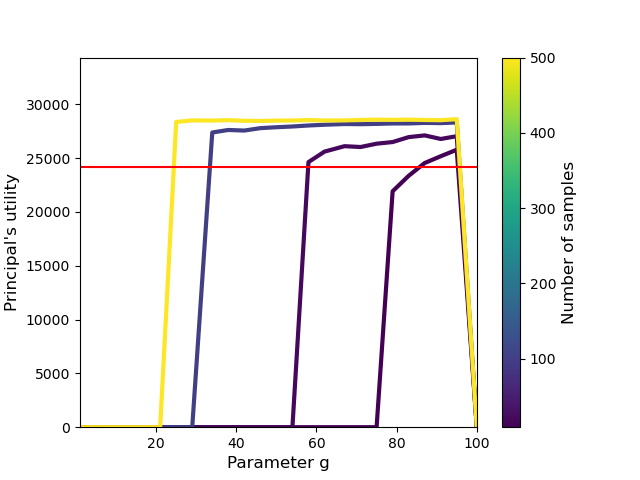}
    \caption{$\bar{\psi} = 6$ and $\pi= 0.9$}
    \label{fig:opt1000.96}
\end{subfigure}
\hfill
\begin{subfigure}{0.32\textwidth}
    \includegraphics[width=\textwidth]{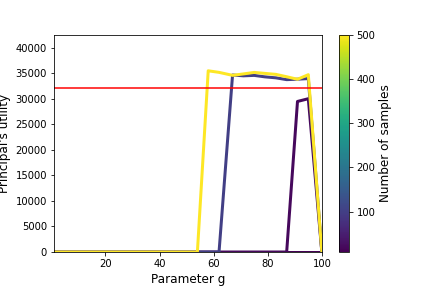}
    \caption{$\bar{\psi} = 8$ and $\pi= 0.3$}
    \label{fig:opt1000.38}
\end{subfigure}
\hfill
\begin{subfigure}{0.32\textwidth}
    \includegraphics[width=\textwidth]{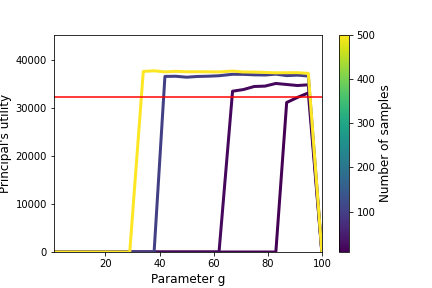}
    \caption{$\bar{\psi} = 8$ and $\pi= 0.3$}
    \label{fig:opt1000.68}
\end{subfigure}
\hfill
\begin{subfigure}{0.32\textwidth}
    \includegraphics[width=\textwidth]{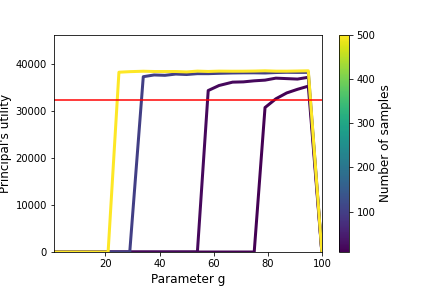}
    \caption{$\bar{\psi} = 8$ and $\pi= 0.9$}
    \label{fig:opt1000.98}
\end{subfigure}
        
\caption{Principal's utility as a function of parameter g}
\label{extraopt}
\end{figure}

\end{appendices}
\end{document}